\documentclass[11pt]{article}
\usepackage[margin=1in]{geometry}

\usepackage[pdftex]{graphicx}
\graphicspath{{../pdf/}{../jpeg/}}
\DeclareGraphicsExtensions{.pdf,.jpeg,.png}
\usepackage{amsmath}
\usepackage{algorithm}  
\usepackage{algorithmic}
\usepackage{array}
\usepackage{url}
\usepackage{textcomp}
\usepackage{amssymb,amsmath,url}
\usepackage{bm} 
\usepackage{multirow}
\usepackage{booktabs}	
\usepackage{epstopdf}
\usepackage{color}
\usepackage{subfigure}
\usepackage{pifont}
\usepackage{tikz}
\usepackage{color}
\usepackage{cite}
\usetikzlibrary{calc}
\usepackage{soul} 

\newtheorem{theorem}{Theorem}
\newtheorem{lemma}{Lemma}
\newtheorem{proposition}{Proposition}
\newtheorem{corollary}{Corollary}

\allowdisplaybreaks
\newtheorem{definition}{Definition} 
\newtheorem{proof}{Proof}

\newtheorem{example}{Example}

\newcommand{\R}{\mathbb{R}}

\newcommand{\norm}[1]{\left\lVert#1\right\rVert}

\newcommand*{\QEDA}{\hfill\ensuremath{\blacksquare}}

\DeclareMathOperator*{\argmin}{arg\,min}

\usepackage{lipsum}

\date{December 31, 2018}
\begin{document}

\title{Designing Coalition-Proof Reverse Auctions\\ over Continuous Goods\let\thefootnote\relax\footnotetext{This research was gratefully funded by the European Union ERC Starting Grant CONENE.}\let\thefootnote\relax\footnotetext{This work has been submitted to the IEEE for possible publication. Copyright may be transferred without notice, after which this version may no longer be accessible.}}
\author{Orcun Karaca\thanks{Orcun Karaca, Pier Giuseppe Sessa and Maryam Kamgarpour are with the Automatic Control Laboratory, Department of Information Technology and Electrical Engineering, ETH Z\"{u}rich, Switzerland. E-mails: {\tt\footnotesize \{okaraca, mkamgar\}@control.ee.ethz.ch, sessap@student.ethz.ch}} \and Pier G. Sessa\footnotemark[1] \and Neil Walton\thanks{Neil Walton is with the School of Mathematics, University of Manchester, United Kingdom. E-mail: {\tt\footnotesize neil.walton@machester.ac.uk}} \and Maryam Kamgarpour\footnotemark[1]
}

\maketitle

\begin{abstract}
This paper investigates reverse auctions that involve continuous values of different types of goods, general nonconvex constraints, and second stage costs. We seek to design the payment rules and conditions under which coalitions of participants cannot influence the auction outcome in order to obtain higher collective utility. Under the incentive-compatible Vickrey-Clarke-Groves mechanism, coalition-proof outcomes are achieved if the submitted bids are convex and the constraint sets are of a polymatroid-type. These conditions, however, do not capture the complexity of the general class of reverse auctions under consideration. By relaxing the property of incentive-compatibility, we investigate further payment rules that are coalition-proof without any extra conditions on the submitted bids and the constraint sets. Since calculating the payments directly for these mechanisms is computationally difficult for auctions involving many participants, we present two computationally efficient methods. Our results are verified with several case studies based on electricity market data.
\end{abstract}
\section{Introduction}\label{sec:1}

{A}{uctions} are effective for allocating resources {and determining their values among a set of} participants. {In an auction,} participants  submit their valuations of the resources (bids), and then {a} central operator determines the allocation and the payment for each participant {as specified by} the auction mechanism. Recently, there has been a surge of interest from the control community to explore auction mechanisms because of their applications in several problems, such as coordinating electric vehicle charging\cite{zou2017efficient}, telecommunication and power networks\cite{zou2017resource}, demand response management \cite{zhou2017eliciting}, provisioning of a distributed database \cite{sanghavi2008new}, and selecting a host for a noxious resource (e.g., trash disposal facility)~\cite{wang2017ex}. 
The central element in any mechanism is the design of the payment and allocation rules, since the participants have incentives to strategize around these rules. In particular, the central operator designs the payment rule to ensure an efficient outcome, that is, an outcome maximizing social welfare. This goal {is best} achieved by solving for the optimal allocation under the condition that the participants submitted their true valuations.

In this paper, we consider mechanism design for reverse auctions that may involve continuous values of different types of goods, general nonconvex constraints, and second stage costs. This is motivated by the fact that several current energy market problems can be cast within this general class of auctions \cite{abbaspourtorbati2016swiss,conejo2010decision, carlson2012miso}. Previous work on these markets considers the pay-as-bid mechanism~\cite{orcun2018game} and the locational marginal pricing mechanism \cite{wu1996folk,schweppe2013spot}. In both mechanisms, participants can bid strategically to influence their {payoff} since these mechanisms do not incentivize truthful bidding. As an alternative, we analyze the Vickrey-Clarke-Groves (VCG) mechanism, which has several theoretical virtues\cite{vickrey1961counterspeculation,clarke1971multipart,groves1973incentives}. 
In particular, unlike many other mechanisms, under the VCG mechanism, every participant finds it more profitable to bid truthfully regardless of the bids of the others. Hence, truthful bidding is the dominant-strategy Nash equilibrium, and we refer to this property as incentive-compatibility. Due to this property, several recent {works propose} the use of the VCG mechanism \cite{xu2017efficient,pgs,samadi2012advanced}.

Despite the desirable theoretical properties of the VCG mechanism, in practice, this mechanism is often deemed undesirable since coalitions of participants can strategically bid to increase their collective utility. Consequently, it is susceptible to different kinds of manipulations, such as collusion and shill bidding \cite{ausubel2006lovely}. 
These shortcomings are decisive in practicality of the VCG~mechanism since in a larger context the VCG~mechanism is not truthful. As a result, practical applications of the VCG~mechanism in real commerce are rare at best~\cite{rothkopf2007thirteen, ausubel2006lovely}.

As is outlined in combinatorial auction literature \cite{ausubel2006lovely}, the shortcomings described above occur when the VCG utilities are not in the~\textit{core}. The core is a concept from coalitional game theory where the participants have no incentives to leave the grand coalition, that is, the coalition of all participants \cite{osborne1994course}. 
 Recently, coalitional game theory has received attention, especially for aggregating power generators~\cite{baeyens2011wind}, and deriving control policies for multi-agent systems~\cite{maestre2014coalitional}. In this paper, we use coalitional game theory to ensure that the VCG mechanism is coalition-proof, in other words, collusion and shill bidding are not profitable. To this end, we derive conditions on the submitted bids and the constraint sets that ensure core VCG utilities by utilizing some recent advances from combinatorial optimization. We show that under convex (or marginally increasing) bids and~polymatroid-type constraints the VCG mechanism is coalition-proof.

The restricted market setting for core VCG utilities, however, does not capture the complexity of general reverse auctions arising in electricity markets. Specifically, these markets may involve nonconvex bids (e.g., generator startup costs\cite{xu2017efficient}), and complex constraint sets that are not polymatroids (e.g., DC or AC optimal power flow constraints~\cite{wu1996folk}). Hence, it may not be possible to ensure core VCG utilities. To this end, we focus on payment rules that are coalition-proof without any extra conditions on the bids and the constraints. 
In particular, we show that selecting the payments from the core yields this desirable property. Naturally, these mechanisms relax the incentive-compatibility of the VCG mechanism. In order to alleviate this issue and incentivize truthfulness, we minimize the participants' abilities to benefit from strategic manipulation among all other coalition-proof mechanisms. 
Variants of these mechanisms are implemented in high-stakes auctions for selling spectrum licenses in the United Kingdom and Switzerland \cite{day2008core,erdil2010new,day2012quadratic,bunz2015faster}. However, these studies are limited to the forward combinatorial auctions. We extend these studies to reverse auctions with potentially continuous values of goods, second stage costs, and general constraints. 


Let us contrast our work with existing mechanism design research.
		The authors in~\cite{lazar2001design} design a forward VCG auction for continuous goods by restricting each participant to submitting a single price-quantity pair to the operator. They show that this mechanism, called progressive second price (PSP) mechanism, has a truthful $\epsilon$-Nash equilibrium, which can be attained by best-response dynamics. The work in~\cite{jia2010analysis} studies the PSP mechanism subject to quantized pricing assumptions, and proposes a fast algorithm converging to a quantized Nash equilibrium with a high probability. More recently, the work in \cite{zou2017efficient} addresses cross-elasticity in PSP design arising from having a multi-period problem for charging electric vehicles. This work is generalized to decentralized procedures and double-sided auctions in \cite{zou2017resource}. In regard to these studies, the PSP mechanism cannot be implemented in dominant-strategies, and truthfulness is only in the price dimension for a given quantity. Specifically, in the PSP mechanism, there is no single true quantity to declare, and the optimal
		quantity depends on the bids of other participants~\cite[\S 3.2]{lazar2001design}. Moreover, the convergence analysis of the best response dynamics is limited to strongly concave true valuations and simple market constraints, for example, availability of a fixed amount of a single continuous good~\cite[\S 3.1]{lazar2001design}.
On the other hand, the work in~\cite{xu2017efficient} applies the VCG mechanism to the wholesale electricity markets and shows that it results in larger payments than the locational marginal pricing mechanism. Nevertheless, none of the aforementioned works consider coalitional manipulations. Finally, the works in \cite{marden2013overcoming,marden2014generalized,li2014decoupling} study the design of the participants' utilities such that the selfish behavior of the participants results in a social welfare maximizing outcome. By contrast, in our case, the true valuations of the participants are \textit{a priori} unknown and they are not part of the design. Instead, we are guiding the participants to a social welfare maximizing outcome by designing meaningful incentives through the payment rule. 

 \vspace{-0.05cm}
Our contributions in the paper are as follows. First, we prove that in the reverse auctions over continuous goods the VCG mechanism is coalition-proof and the VCG utilities lie in the core, if and only if the market objective function~is supermodular, see Theorems~2~and~3. These results are direct extensions of the results from the multiple item setting. Second, considering the special setting of continuous goods and complex constraints, we derive novel conditions on the bids and the constraint sets under which the VCG mechanism is coalition-proof, see Theorems~4~and~5. Third, in Theorem~6, we show that in the reverse auctions over continuous goods selecting payments from the core results in a coalition-proof mechanism without any restrictions on~the bids and the constraints, extending results from the multiple~item~setting. Finally, we address the computational difficulties of these mechanisms by providing two efficient methods. 

The remainder of this paper is organized as follows. In Section~\ref{sec:2}, we introduce a constrained optimization problem that models a general class of reverse auctions. Then, we introduce the VCG mechanism and illustrate that coalitions of participants can influence the outcome. For an analysis of this shortcoming, in Section~\ref{sec:3} we bring in tools from coalitional game theory, namely the {core}, in which the participants do not have incentives to leave the grand coalition. Throughout this section, we investigate conditions under which the VCG mechanism is coalition-proof. Since these conditions do not capture the complexity of the general class of reverse auctions, alternative payment rules are proposed in Section~\ref{sec:4}. In Section~\ref{sec:5}, we present several case studies based on real-world electricity market data to illustrate the conditions we derived and the proposed payment~rules.

\section{Mechanism Framework}\label{sec:2}
We start with a generic reverse auction model. The set of auction participants consists of the central operator $l=0$ and the bidders $l\in L$, where  $L=\{1,\ldots,\lvert L\rvert\}$. Let $t$ be the number of types of goods in the reverse auction. Goods of the same type from different bidders are fungible to the central operator. Each bidder~$l$ has a private true cost~function $c_l:\mathbb R_+^t \rightarrow \mathbb R_+$ which is nondecreasing. We further assume that~$c_l(0)=0$. This assumption holds for many existing reverse auctions, for instance, control reserve markets and day-ahead electricity markets that include generators' start-up costs~\cite{abbaspourtorbati2016swiss, xu2017efficient}. Each bidder~$l$ then submits a bid function to the central operator, denoted by $b_l:\mathbb R_+^t \rightarrow \mathbb R_+$ and nondecreasing with $b_l(0)=0$.

Given the bid profile $\mathcal{B}=\{b_l\}_{l\in L}$, \textit{a mechanism} defines an allocation rule $x_l^*(\mathcal{B})\in\R^{t}_+$ and a payment rule $p_l(\mathcal{B})\in\R$ for each bidder $l$. 
The central operator's objective is 
\begin{equation*}
\bar{J}(x,y;\mathcal{B})=\sum\limits_{l\in L} b_l(x_l) + d(x,y).
\end{equation*}
Here $y\in\R^p$ are additional variables entering the central operator's optimization, in addition to the allocation $x\in\R^{t\lvert L\rvert}$. The function $d:\R^{t\rvert L\rvert}_+\times\R^p\rightarrow \R$ is an additional cost term. For example, in a two-stage auction model, the operator can buy the goods from another market. In this case, $y$ corresponds to the second stage variables and the function $d$ is the second stage cost \cite{abbaspourtorbati2016swiss}. In Section~\ref{sec:5}, we provide a real-world electricity market example where the function $d$ incorporates expected daily market prices in a weekly market.  

In most auction mechanisms, the allocation is determined by minimizing the central operator's objective subject to some market constraints
\begin{equation}\label{eq:main_abstraction}
J(\mathcal{B})=\min_{x,y}\, \bar{J}(x,y;\mathcal{B})\ \,\mathrm{s.t.}\ \,g(x,y)\leq 0, 
\end{equation}
where $g:\R^{t\rvert L\rvert}_+\times\R^{p}\rightarrow \R^{q}$ defines the constraints. (If the above optimization problem is infeasible then the objective value is $J(\mathcal{B})=\infty$.) 
Let the optimal solution of \eqref{eq:main_abstraction} be denoted by $x^*(\mathcal{B})\in \R^{t\rvert L\rvert}_+$ and $y^*(\mathcal{B})\in\R^{p}$. We assume that in case of multiple optima, there is some tie-breaking rule according to a predetermined fixed ordering of the bidders. The solution of the optimization problem \eqref{eq:main_abstraction} is the optimal allocation with respect to the submitted bids, that is, the goods are bought from the bidders with lower bid prices. 

The \textit{utility} of bidder $l$ is given by \begin{equation*}
u_l(\mathcal{B})=p_l(\mathcal{B})-c_l(x^*_l(\mathcal{B})).
\end{equation*} 
A bidder whose bid is not accepted, $x_l^*(\mathcal{B})=0$, is not paid and $u_l(\mathcal{B})=0$.
The total payment made by the central operator is given by  
\begin{equation*}
u_0(\mathcal{B})=-\sum\limits_{l\in L} p_l(\mathcal{B}) - d(x^*(\mathcal{B}),y^*(\mathcal{B})),
\end{equation*}
which can be treated as the utility of the central operator.  Note that this payment can be an expected payment when the function $d$ is an expected second stage cost. As a remark, if the optimization problem \eqref{eq:main_abstraction} is infeasible, we have $u_0(\mathcal{B})=-\infty$.

Constrained optimization problem (\ref{eq:main_abstraction}) defines a general class of reverse auction models. Several current market problems such as stochastic energy market mechanisms \cite{abbaspourtorbati2016swiss,conejo2010decision,pgs,orcun2018game,bouffard2005market,bouffard2005market2,ahmadi2014multi} and energy reserve co-optimized markets \cite{xu2017efficient, carlson2012miso, cheung1999energy, chow2005electricity,amjady2009stochastic,kargarian2014spider,reddy2015joint} can be cast within this model.\footnote{Let us contrast our market framework with the one in \cite{pritchard2010single}. The authors in \cite{pritchard2010single} study a single settlement mechanism for two-stage stochastic markets. These markets include intermittent power generation from renewable resources such as solar energy and wind. These resources cannot accurately predict the quantity they can produce in the first stage of the market. These productions are revealed only in real-time. To accurately price such generation, the mechanism in \cite{pritchard2010single} ties the first stage and the second stage markets together. To this end, the bidders are required to provide linear bid functions for both stages of the market. The mechanism in~\cite{pritchard2010single} then sets prices for both stages of the market simultaneously. The main difference to our work is that the mechanisms we study distribute payments only for the first stage of the market. We assume that the settlement of the second stage is separate from that of the first stage motivated by the markets in \cite{abbaspourtorbati2016swiss, conejo2010decision}.} The constraints may correspond to procurement of the required amounts of power supplies, for instance, in the Swiss control reserve markets accepted reserves must have a deficit probability of less than 0.2\%.  They may also correspond to a transmission network, for instance, in energy markets power injections must satisfy the transmission line limits and the power flow equations. 

Three fundamental properties we desire in mechanism design are 1) individual rationality, 2) efficiency and 3) dominant-strategy incentive-compatibility. A mechanism is \textit{individually rational} if the bidders do not face negative utilities, $u_l(\mathcal{B})\geq 0$, for all $l\in L$. A mechanism is \textit{efficient} if the sum of all the utilities $\sum_{l=0}^{\lvert L\rvert} u_l(\mathcal{B})$ is maximized. 

To define the third property, we first bring in tools from game theory. Let $\mathcal{B}_{-l}$ be the bid profile of all the bidders, except bidder $l$.
The bid profile $\mathcal{B}$ is a \textit{ Nash equilibrium} if for every bidder~$l$, $u_l(\mathcal{B}_l,\mathcal{B}_{-l})\geq u_l(\tilde{\mathcal{B}}_l,\mathcal{B}_{-l})$, $\forall\tilde{\mathcal{B}}_l$.
The bid profile $\mathcal{B}$ is a \textit{dominant-strategy Nash equilibrium} if for every bidder $l$,
$u_l(\mathcal{B}_l,\hat{\mathcal{B}}_{-l})\geq u_l(\tilde{\mathcal{B}}_l,\hat{\mathcal{B}}_{-l})$, $\forall\tilde{\mathcal{B}}_l$, $\forall \hat{\mathcal{B}}_{-l}$. 
Let the truthful bid profile be $\mathcal C=\{c_l\}_{l\in L}$. Then, a mechanism is  \textit{dominant-strategy incentive-compatible} if the bid profile $\mathcal C$ is the dominant strategy Nash equilibrium. In other words, every bidder finds it more profitable to bid truthfully, regardless of the bids of other participants.

\subsection{Currently used payment rules}\label{sec:2a}
We introduce two prominent payment rules widely used for the reverse auctions under consideration.

In the \textit{pay-as-bid mechanism}, the payment rule is
$p_l(\mathcal{B})=b_l(x^*_l(\mathcal{B}))$.\footnote{For instance, several European balancing market are settled under a pay-as-bid
	mechanism---see~\cite{mazzi2018price} and further references therein.}
It follows that each bidder's utility is $u_l(\mathcal{B})=b_l(x^*_l(\mathcal{B}))-c_l(x^*_l(\mathcal{B}))$. A rational bidder would overbid to ensure positive utility. Consequently, under the pay-as-bid mechanism, the central operator calculates the optimal allocation for the inflated bids rather than the true costs. Furthermore, the bidders need to spend resources to learn how to bid to maximize their utility. There are many Nash equilibria arising from the pay-as-bid mechanism, none of which are incentive-compatible. This result was previously shown in our work, see \cite{orcun2018game}.

The \textit{locational marginal pricing (LMP) mechanism} is adopted in the energy markets where transmission networks are present. We refer to \cite{wu1996folk} for an exposition on the calculation of these payments from the Karush-Kuhn-Tucker conditions of the optimization problem \eqref{eq:main_abstraction}. Under this mechanism, a bidder can manipulate its LMP payment by inflating its bids. As a strategic manipulation, a bidder may also withhold its maximum supply\cite{ausubel2014demand,wolfram1997strategic,joskow2001quantitative}.  Similar to the pay-as-bid mechanism, the bidders need to spend resources to learn how to maximize their utilities. Furthermore, a Nash equilibrium may not exist \cite{TangJ13, madrigal2001existence}.

We see that these payment rules do not satisfy the properties of efficiency or incentive-com\-patibility.~Next,~we~explore the properties of the Vickrey-Clarke-Groves mechanism.

\subsection{Vickrey-Clarke-Groves mechanism}\label{sec:2b}

The \textit{Vickrey-Clarke-Groves (VCG) mechanism} is characterized with the payment rule:
\begin{equation*}\label{eq:vcg_definition}
p_l(\mathcal{B})=b_l(x^*_l(\mathcal{B}))+(h(\mathcal{B}_{-l})-J(\mathcal{B})).
\end{equation*}
The function $h(\mathcal{B}_{-l})\in\R$ must be chosen carefully to ensure individual rationality. A particular choice is the \textit{Clarke pivot rule} $
h(\mathcal{B}_{-l})=J(\mathcal{B}_{-l}),  $
where $J(\mathcal{B}_{-l})$ denotes the minimum total cost without bidder $l$, that is, the optimal value of the optimization problem (\ref{eq:main_abstraction}) with $x_l=0$.\footnote{In the case of an auction of a single item, the VCG mechanism with the Clarke pivot rule is equivalent to the second-price (Vickrey) auction.} Note that this mechanism is well-defined under the assumption that there exists a feasible solution when a bidder is removed. However, this is not a restrictive assumption in the presence of many bidders and a second-stage market. 

For the model introduced in \eqref{eq:main_abstraction}, our first result shows that the VCG mechanism first derived in  \cite{vickrey1961counterspeculation,clarke1971multipart,groves1973incentives} satisfies all three fundamental properties. This result is a straightforward generalization of the works in \cite{vickrey1961counterspeculation,clarke1971multipart,groves1973incentives}, which do not consider continuous values of goods, second stage costs, and general constraints.

\begin{theorem}
	\label{thm:incentive_comp}
	Given the market model \eqref{eq:main_abstraction}, 
	\begin{enumerate}
		\item[(i)] The VCG mechanism is {dominant-strategy} {incentive-compatible}. 
		\item[(ii)] The VCG mechanism is \text{efficient}.
		\item[(iii)] The VCG mechanism ensures nonnegative payments and \text{individual rationality} when the Clarke pivot rule is utilized, $
		h(\mathcal{B}_{-l})=J(\mathcal{B}_{-l})$.
	\end{enumerate}
\end{theorem}

The proof is relegated to Appendix~\ref{app:A}. In summary, all bidders have  incentives to reveal their true costs in a VCG mechanism. Dominant-strategy incentive-compatibility makes it easier for entities to enter the auction, without spending resources in computing optimal bidding strategies. This can promote participation in the market. As a remark, Theorem~\ref{thm:incentive_comp}-(ii) shows that solving for the optimal allocation with the true costs yields an efficient mechanism. In the remainder, we consider the Clarke pivot rule for the VCG mechanism since it ensures individual rationality.

Despite many advantages of the VCG mechanism for the reverse auction model (\ref{eq:main_abstraction}), this mechanism suffers from collusion and shill bidding~\cite{ausubel2006lovely}. Bidders $K\subseteq L$ are \textit{colluders} if they obtain higher collective utility by changing their bids from $\mathcal C_K=\{c_l\}_{l\in K}$ to $\mathcal{B}_K=\{b_l\}_{l\in K}$. Bidder $l$ is a \textit{shill bidder} if there exists a set $S$ and bids $\mathcal{B}_S=\{b_k\}_{k\in S}$ such that bidder~$l$ finds participating with bids~$\mathcal{B}_S$ more profitable than participating with a single truthful bid~$\mathcal{C}_l$. 

To illustrate these issues, we study two energy market examples. For these examples, we consider the VCG mechanism with the Clark-pivot rule and with $d(x,y)\equiv 0$ in the central operator's objective.  Later, we come back to these examples in Section \ref{sec:3}, in order to discuss conditions to eliminate collusion and shill-bidding. First example is a reverse auction of a single type of power supply. In these markets, each bidder is allowed to submit mutually exclusive bids that can equivalently be represented as bid curves, see Section~\ref{sec:3a}.

\begin{example}[Simple Market]\label{ex:first_simple_example}
	Suppose the central operator has to procure $800$ MW of power supply from bidders $1$, $2$ and $3$ who have the true costs $\$100$ for $400$ MW, $\$400$ for $400$ MW and $\$600$ for $800$ MW, respectively. Under the VCG mechanism, bidders $1$ and $2$ win and receive $p_1^{\text{VCG}} = 100 + (600-500)= \$200$ and $p_2^{\text{VCG}} = 400 + (600-500)= \$500$. Suppose bidders $1$ and $2$ collude and change their bids to $\$0$ for $400$ MW. Then, bidders $1$ and $2$ receive a payment of $\$600$ each for the same allocation. In fact, bidders~$1$~and~$2$ could represent multiple identities of a single losing bidder (that is, a bidder with the true cost greater than $\$600$ for $800$ MW). Entering the market with two shill bids, this bidder receives a payment of $2\times\$600$ for $800$ MW.
\end{example}
\vspace{.1cm}

Second example is a power market where the operator is procuring a set of different types of power supplies.

\begin{example}[Power Market]\label{ex:second_simple_example}
	We consider a reverse auction with three different types of supplies, types A, B and C. Here, type A can replace types~B~and~C simultaneously.\footnote{This is an abstraction for power reserve markets where secondary reserves can replace both negative and positive tertiary reserves simultaneously \cite{abbaspourtorbati2016swiss}.} Suppose the central operator has to procure $100$ MW of type B and $100$ MW of type C (or equivalently only $100$ MW of type~A) from bidders $1$ to $5$. Truthful bid profiles of $5$ bidders are provided in Table~\ref{tab:bidprof}. 
	\begin{table}[h]
		\caption{Bid profile in the Power Market}
		\label{tab:bidprof}
		\begin{center}
			\begin{tabular}{|l||l||l||l||l||l|}
				\hline
				Bidders (Types) & $1$ (A) & $2$ (B) &  $3$ (B) & $4$ (C) & $5$ (C) \\
				\hline 
				MW & $100$ & $100$  & $100$ & $100$& $100$ \\
				\hline
				\$ & $500$  & $350$  & $400$ & $250$ & $400$ \\
				\hline
			\end{tabular}
		\end{center}
	\end{table}
	
	Constraint set in \eqref{eq:main_abstraction} is given by
	\begin{equation}	\label{newc2}
	\begin{split}
	\big\{ x\in\{0,100\}^5\, \rvert\, & x_1 + x_2 + x_3 \geq M(\{\text{A},\text{B}\})=100,\\
	& x_1 + x_4 + x_5 \geq M(\{\text{A},\text{C}\})=100\big\}.  
	\end{split}
	\end{equation}	
	Under the VCG mechanism, bidder $1$ wins and receives $p_1^{\text{VCG}} = 500 + (600-500)= \$600$. Suppose losing bidders $2$ and $4$ collude and change their bid prices to $\$0$. Then, bidders $2$ and $4$ receive $\$400$ each and they obtain a collective VCG profit of $\$200$. Total payment of the operator increases from $\$600$ to $\$800$. This is unfair towards bidder $1$ who is willing to offer the same supply for $\$500$.
\end{example}
\vspace{.1cm}

It is troubling that the VCG mechanism can result in large payments through coalitional manipulations. In all these examples, there exists a group of bidders who is willing to offer the same amount of good by receiving less payment. From the central operator's perspective, the operator would instead want to renegotiate the payments with a subset of participants.

Next, we define the desirable auction outcomes as the \textit{coalition-proof} outcomes. By coalition-proof, we mean that a group of bidders who lose when bidding their true cost, cannot profit by a joint deviation, and a bidder cannot profit from bidding with multiple identities. Consequently, auctions with coalition-proof outcomes are immune to collusion and shill bidding, and they would avoid large payments through coalitional manipulations. We remark that it is not possible to expect being fully immune to collusion from all sets of bidders. For instance, no mechanism can eliminate the case where all bidders inflate their bid prices simultaneously, see also the examples in~\cite{beck2009revenue}. Hence, we concentrate our efforts on eliminating collusion from the bidders who lose when bidding truthfully.

\section{Ensuring Coalition-Proof VCG Outcomes}\label{sec:3}

In coalitional game theory, the \textit{core} defines the set of utility allocations that cannot be improved upon by forming coalitions.\footnote{Throughout the paper, we use the term utility allocation and the auction outcome interchangeably.} Here, we show that if the VCG outcome lies in the core, then the VCG mechanism eliminates {{any}} incentives for {collusion} and {shill bidding}. Keeping this in mind, our goal is to derive sufficient conditions to ensure that the VCG outcome lies in the core, and hence the VCG mechanism is coalition-proof.

For every $S\subseteq L$,  let $J(\mathcal{B}_S)$ be the objective function under any set of bids $\mathcal{B}_S=\{b_l\}_{l\in S}$ from the coalition $S$. It is defined by the following expression:
\begin{equation}
\label{eq:33}
\begin{split}
J(\mathcal{B}_S) =  &\min_{x,y}\sum\limits_{l\in S} b_l(x_l) + d(x,y)\\
&\ \mathrm{  s.t. } \ g(x,y)\leq 0,\, x_{-S}=0,
\end{split}
\end{equation}
where the stacked vector $x_{-S}\in\R_+^{t(\lvert L\rvert -\lvert S\rvert)}$ is defined by omitting the subvectors from the set $S$. Note that this function is nonincreasing, that is, $J(\mathcal{B}_R) \geq J(\mathcal{B}_S)$ for $R\subseteq S$.\footnote{This holds since $b_l(0)=0$ for all $l\in L$.} 

Next, we define the core with respect to the truthful bids, $\mathcal{C}_R=\{c_l\}_{l\in R}$, and refer to this definition solely as the \textit{core}. 
\begin{definition}\label{def:core_def}
	For every set of bidders $R\subseteq L$, the core $Core(\mathcal{C}_R)\in\R\times\R^{\rvert R\rvert}_+$ is defined as follows
	\begin{equation}\label{eq:mcoredef}
	Core(\mathcal{C}_R)=\Big\{u\in\R\times\R^{\rvert R\rvert}_+ \,|\, u_0+\sum\limits_{l\in R} u_l=-J(\mathcal{C}_R),\,
	u_0+\sum\limits_{l\in S}u_l\geq-J(\mathcal{C}_S),\, \forall S \subset R \Big\}.
	\end{equation}
\end{definition}
Note that there are $2^{\rvert R\rvert}$ linear constraints that define a utility allocation in the core for the set of bidders $R$. The core is always nonempty in an auction because the utility allocation $u_0=-J(\mathcal{C}_R)$  and $u_l=0$ for all $l\in R$ always lies in the core. This allocation corresponds to the utility allocation of the pay-as-bid mechanism under the truthful bidding $\mathcal{C}_R$.

For the mechanism design, we highlight the implications of the constraints in~\eqref{eq:mcoredef}. Restricting the utility allocation to the nonnegative orthant yields the individual rationality property for the bidders. The equality constraint implies that the mechanism is efficient, since the term on the right is maximized by the optimal allocation. We say that a utility allocation is unblocked if there is no set of bidders that could make a deal with the operator from which every member can benefit, including the operator. This condition is satisfied by the inequality constraints. 

The VCG outcome attains the maximal utility in the core for every bidder. Note that under the VCG mechanism each bidder's utility is given by $u_l^{\text{VCG}}=J(\mathcal{C}_{-{l}})-J(\mathcal{C})$. Then, for every bidder~$l$, $u_l^{\text{VCG}}=\max\left\{u_l\,\rvert\, u\in Core(\mathcal{C})\right\}$, see \cite{ausubel2002ascending,orcun2018game}. {In general, this maximal point may not lie in the core. The~following example gives visual insight about the core in terms of payments and illustrates the dominant-strategy Nash equilibrium of the VCG mechanism corresponding to Example~\ref{ex:first_simple_example}. In this example, shill bidding and collusion are profitable under the VCG mechanism. 
	\begin{example}\label{ex:core_illust}
		We revisit Example \ref{ex:first_simple_example}. Without loss of generality, assume that in case of a tie the central operator prefers bidders $1$ and $2$ over bidder~$3$. We can visualize the core outcomes for the bidders $1$ and $2$ by removing the losing bidder~$3$, $p_3^{\text{VCG}} = 0$, and the operator. Core outcomes and the VCG payments ($p_i^{\text{VCG}}$) are given in Figure \ref{fig:core_picture}. 
		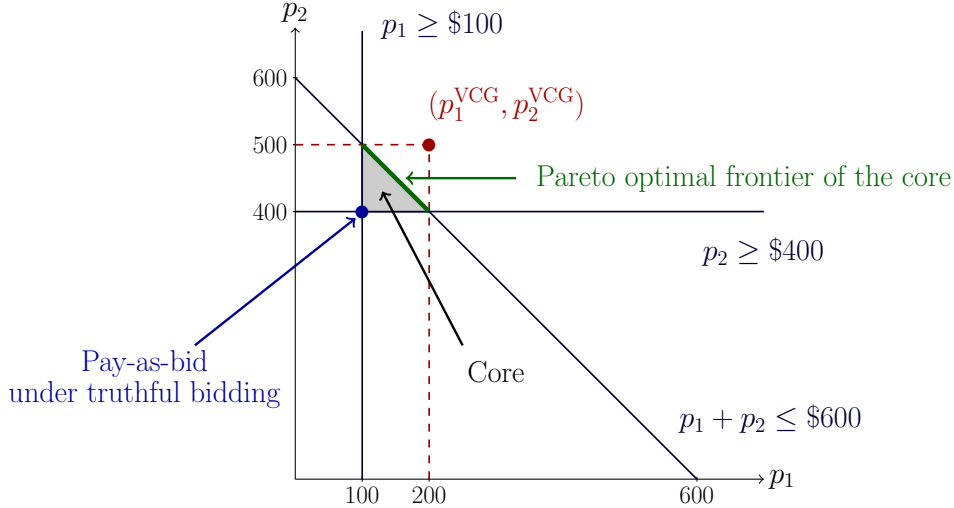
\begin{figure}[th]
			\begin{center}
				\begin{tikzpicture}[scale=0.89, every node/.style={scale=0.55}]
				\coordinate (r0) at (1,5);
				\coordinate (s0) at (2,4);
				\coordinate (si) at (1,4);
				\coordinate (s1) at (1,5);
				\draw[->] (0,0) -- (7,0) node[right] {\huge $p_1$};
				\draw[->] (0,0) -- (0,6.75) node[above] {\huge $p_2$};
				\foreach \x in {1,2,6}
				\draw (\x cm,1pt) -- (\x cm,-1pt) node[anchor=north] {\LARGE\pgfmathparse{100*\x} \pgfmathprintnumber[    
					fixed,
					fixed zerofill,
					precision=0
					]{\pgfmathresult}};
				\foreach \y in {4,5,6}
				\draw (1pt,\y cm) -- (-1pt,\y cm) node[anchor=east]  {\LARGE\pgfmathparse{100*\y} \pgfmathprintnumber[    
					fixed,
					fixed zerofill,
					precision=0
					]{\pgfmathresult}};;
				\draw[scale=1, line width=.22mm, domain=0:6,smooth,variable=\x,black!80!blue] plot ({\x},{6-\x}) node[below] at(7.1,1.2){\huge $p_1+p_2\leq \$600$};
				\draw[scale=1,line width=.22mm, domain=0:7,smooth,variable=\x,black!80!blue]  plot ({\x},{4}) node[right] at(6,3.4) {\huge $p_2\geq \$400$};
				\draw[scale=1, line width=.22mm,domain=0:6.7,smooth,variable=\y,black!80!blue]  plot ({1},{\y}) node[above] at(2.2,6.5) {\huge $p_1\geq  \$100$};
				\draw[scale=1, line width=.22mm, domain=0:2,dashed,variable=\x,black!40!red]  plot ({\x},{5});
				\draw[scale=1, line width=.22mm, domain=0:5,dashed,variable=\y,black!40!red]  plot ({2},{\y}) node[right]  at(1.9,5.6){\huge $(p_1^{\text{VCG}},p_2^{\text{VCG}})$};
				\filldraw[draw=black!80!blue,line width=.22mm, fill=gray!40] (r0) -- (s0) -- (si) -- (s1) -- cycle;
				\node[black!40!red] at (2,4.975) {\Huge\textbullet};
				\node[black!40!blue] at (1,3.975) {\Huge\textbullet};
				\draw[->,line width=.32mm,black!60!green] (3.3,4.5) -- (1.65,4.5) ;
				\node[black!60!green]  at (6.7,4.5) {\huge Pareto optimal frontier of the core};
				\draw[->,line width=.32mm] (2.5,2) -- (1.3,4.3) ;
				\node at (3,1.6) {\Huge $\substack{\text{Core}}$};
				\draw[->,line width=.32mm,black!40!blue] (-1.5,2) -- (0.9,3.9) ;
				\node[black!40!blue] at (-2.25,1.5) {\Huge $\substack{\text{Pay-as-bid}\\ \text{under truthful bidding}}$};
				\draw[scale=1, line width=.6mm, domain=1:2,smooth,variable=\x,black!60!green] plot ({\x},{6-\x});
				\end{tikzpicture}
				\caption{Core outcomes and the VCG payments under truthful bidding}\label{fig:core_picture}
			\end{center}
		\end{figure}
		\end{example}

	Considering that the core will characterize coalition-proof outcomes, we are ready to investigate the conditions under which the VCG outcome lies in the core. To this end, we {provide three} sufficient conditions that ensure core VCG outcomes for the auction model~(\ref{eq:main_abstraction}). Notice that there are $2^{|L|}$ linear core constraints in $Core(\mathcal{C})$, see~\eqref{eq:mcoredef}. First, we derive the following equivalent characterization with significantly lower number of constraints.
\begin{lemma}
	\label{lem:lemma_core}
	Let $W\subseteq L$ be the winners of the reverse auction \eqref{eq:main_abstraction} for the set of bidders $L$, that is, each bidder $l\in W$ is allocated a positive quantity. Let $ u\in\R\times\R^{\rvert L\rvert}_+$ be the corresponding utility allocation. Then, $ u \in Core(\mathcal{C})$ if and only if $ u_0=-J(\mathcal{C})-\sum_{l\in L} u_l$ and
	\begin{equation}
	\label{eq:core_constraints}
	\sum_{l \in K }  u_l\leq J(\mathcal{C}_{-K}) - J(\mathcal{C}),\ \forall K\subseteq W.
	\end{equation}
\end{lemma}

 The proof is relegated to Appendix~\ref{app:B}. The following proposition is our first sufficient condition for core VCG outcomes.

\begin{proposition}\label{lem:removal_of_two}
	The VCG outcome is in the core, $ u^{\text{VCG}} \in Core(\mathcal{C})$, if the market in \eqref{eq:main_abstraction} is infeasible whenever any {two} winners $l_1,l_2\in W$ are removed from the set of bidders~$L$.
\end{proposition}

\begin{proof} Notice that if the optimization problem \eqref{eq:main_abstraction} is infeasible then the objective value is $J(\mathcal{B})=\infty$. Given this property of the market, inequality constraints \eqref{eq:core_constraints} in Lemma~\ref{lem:lemma_core} simplify to constraints on the utilities of single bidders, that is,
	$u^{\text{VCG}}_l\leq J(\mathcal{C}_{-l}) - J(\mathcal{C}),$ for all $l\in W$.
	This inequality follows directly from the definition of the VCG utility, $u^{\text{VCG}}_l= J(\mathcal{C}_{-l}) - J(\mathcal{C}),$ for all $l\in W$. Equality constraint in Lemma~\ref{lem:lemma_core} is satisfied by definition. \QEDA
\end{proof}

The above condition can only be present in some specialized instances of reverse auctions. It is not possible in general to guarantee that this condition will hold in a market by enforcing restrictions on the bidders and the market model.

{We soon show that supermodularity provides an equivalent condition for core VCG outcomes. }

\begin{definition}
	A function $f:2^{L}\rightarrow\R$ is \textit{supermodular} if $$f(S)-f(S_{-l})\leq f(R)-f(R_{-l}),$$
for all $S\subseteq R\subseteq L$ and for all $l\in S$. Or, equivalently, for all~$S, R\subseteq L$, $f(S\cup R)+f(S\cap R)\geq f(S)+f(R)$ must hold. A function $f:2^{ L}\rightarrow\R$ is \textit{submodular} if $-f$ is supermodular. Furthermore, a function is nondecreasing if $f(S')\leq f(S)$, for all $S'\subseteq S.$
\end{definition}

  For the remainder, the objective function $J$ in~\eqref{eq:33} is said to be supermodular if supermodularity condition holds under any bid profile. Our main result of this section proves that supermodularity of the objective function is necessary and sufficient for ensuring core VCG outcomes.

{
	\begin{theorem}\label{thm:iff_supermodularity}	For any bid profile $\mathcal{C}$ and for any set of participating auction bidders $R\subseteq L$, the outcome of the VCG mechanism is in the core, \text{if and only if} the objective function $J$ in~\eqref{eq:33} is supermodular.  
	\end{theorem} 
}
Note that a similar result was proven in  \cite{ausubel2006lovely}, for a forward auction of a finite number of items without any other constraints. Our result generalizes this proof to reverse auctions with continuous goods and arbitrary central operator objectives and constraints as modeled in \eqref{eq:main_abstraction}. The proof developed also significantly simplifies the arguments in \cite{ausubel2006lovely}. The proof is relegated to Appendix~\ref{app:C}.

As previously anticipated, we now prove that core outcomes, hence the supermodularity condition, makes collusion and shill bidding unprofitable in a VCG mechanism.

\begin{theorem}
	\label{thm:no_collusions}
			For the set of bidders $L$, consider a VCG auction mechanism modeled by \eqref{eq:main_abstraction}. If the objective function $J$ is supermodular, then, 
	\begin{itemize}
	\item[(i)] {A group of bidders who lose when bidding their true values cannot profit by a joint deviation.}
	\item[(ii)] Bidding with multiple identities is unprofitable for any bidder.\end{itemize}
\end{theorem}

The proof is relegated to Appendix~\ref{app:D}.  Theorem~\ref{thm:no_collusions} shows that if the operator has a supermodular objective (or equivalently, if the VCG outcomes are in the core), then the VCG mechanism is coalition-proof. Given this result, we next investigate sufficient conditions on the bids and the constraint sets {in order to ensure supermodularity and thus coalition-proof~outcomes.} In the following we derive conditions for two classes of markets.

\subsection{Markets for a single type of good}\label{sec:3a}
We start by considering simpler reverse auctions where the operator has to procure a fixed amount $M$  of a single type of good. 
Each bidder~$l$ has a private true cost function $c_l:\mathbb R_+ \rightarrow \mathbb R_+$ that is nondecreasing with $c_l(0)=0$. These {types} of auctions are mainly characterized by single-stage decisions with mutually exclusive bids. This means that a bidder can offer a set of bids, of which only one can be accepted. We first show that such discrete bids fit into our model~\eqref{eq:main_abstraction}.  Here, bidder~$l$ submits truthful bids for $n_l$ discrete amounts as $\{(c_{l,i},x_{l,i})\}_{i=1}^{n_l}$
{where $c_{l,i}\in\R_+$ and the amounts offered by each bidder $x_{l,i} \in \R_+$
	must be equally spaced by some increment $m$  which is a divisor of $M$, that is,
	\begin{equation}\label{eq:simpler_clearing_model_cond}
	x_{l,i}= im,\; \text{for some }{i\in\mathbb Z_+}.
	\end{equation}
}Note that, there is an equivalent representation of the form $c_l(x)\in\R_+$ for $x>0$ as follows: 
{	\begin{equation}\label{cdef}
	c_l(x)= \min_{i=1,...,n_l} \left\{ c_{l,i} \,\rvert\, x_{l,i} \geq x\right\},
	\end{equation}
}where $c_l(0)=0$. This form equivalently represents that all the amounts up to the size of the winning bid are available to the operator. Furthermore, bid prices of this form are piecewise constant and continuous from the left.

We consider auctions cleared by
\begin{equation}	\label{eq:simpler_clearing_model}
\begin{split}
J(\mathcal{C}_S) =  &\min_{x\in\R_+^{\rvert S \rvert}}  \;\sum_{l \in S}c_l(x_l)\\ &\ \ \mathrm{s.t. }\ \,   \sum_{l \in S}x_l \geq M,
\end{split}
\end{equation}
{for} $S\subseteq L$.} 
Note, we can equivalently assume that $x_l$, above, takes values in $\{x_{l,i} \,\rvert\, i \in \mathbb{Z_+} \}\subseteq \mathbb{R}_+$ (cf. \eqref{cdef} and Figure \ref{fig:marginally_inc}) and doing so, we let $x^*=\{x^*_l\}_{l \in S}$ be the optimal values in these sets.

The model \eqref{eq:simpler_clearing_model} is within the auction model \eqref{eq:main_abstraction}. We can now derive conditions on bidders' true costs to ensure supermodularity of $J$. Thus, we derive conditions under which the VCG outcome from~\eqref{eq:simpler_clearing_model} would lie in the core. 
\begin{theorem}
	\label{thm:conditions_on_bids}
	Given \eqref{eq:simpler_clearing_model_cond}, if the true costs are marginally increasing, namely,
	$x_{l,b} - x_{l,a} = x_{l,d} - x_{l,c}$
	implies that 
	$ c_{l,b} - c_{l,a} <  c_{l,d} - c_{l,c} $
	for each bidder $l \in L$ and for each $0 \leq x_{l,a}<x_{l,c}< x_{l,d}$, 
	then the objective function $J$ is supermodular.
\end{theorem}

This setting includes reverse auctions of multiple
identical items as a subset. However, we highlight that our proof does not share any similarities with that of \cite{ausubel2006lovely}, which achieves submodular objective functions in forward combinatorial auctions. Our proof relies on an important lemma we prove showing that the allocations of every bidder is nondecreasing when a bidder is removed from the auction \eqref{eq:simpler_clearing_model}. The proof is relegated to Appendix~\ref{app:E}.
As a corollary of this result, marginally increasing costs imply coalition-proof VCG outcomes for~\eqref{eq:simpler_clearing_model}\, and thus eliminate incentives for collusion and shill bidding. This condition is visualized in Figure~\ref{fig:marginally_inc}. 

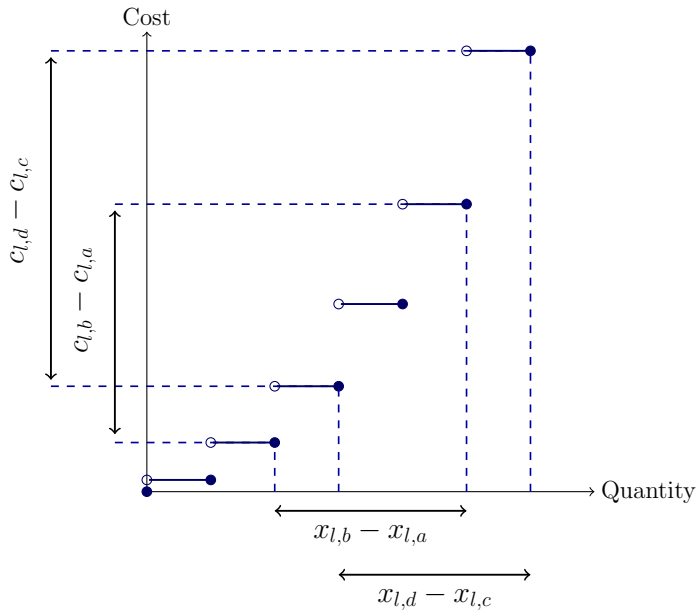
\begin{figure}[t]
	\begin{center}
		\begin{tikzpicture}[scale=0.85, every node/.style={scale=0.65}]
		\draw[->] (0,0) -- (7,0) node[right] {\Large Quantity};
		\draw[->] (0,0) -- (0,7.2) node[above] {\Large Cost};
		\draw[black!60!blue,fill=black!60!blue] (0,0) circle (.075cm);
		\draw[black!60!blue,fill=black!60!blue]  (1,0.1837) circle (.075cm);
		\draw[black!60!blue,fill=black!60!blue]  (2,0.769) circle (.075cm);
		\draw[black!60!blue,fill=black!60!blue]  (3,1.65) circle (.075cm);
		\draw[black!60!blue,fill=black!60!blue] (4,2.937) circle (.075cm);
		\draw[black!60!blue,fill=black!60!blue]  (5,4.5)circle (.075cm);
		\draw[black!60!blue,fill=black!60!blue] (6,6.9) circle (.075cm);
		\draw[black!60!blue]  (0,0.1837) circle (.075cm);
		\draw[black!60!blue]  (1,0.769) circle (.075cm);
		\draw[black!60!blue]  (2,1.65) circle (.075cm);
		\draw[black!60!blue] (3,2.937) circle (.075cm);
		\draw[black!60!blue]  (4,4.5) circle (.075cm);
		\draw[black!60!blue]  (5,6.9) circle (.075cm);
		\draw[scale=1, line width=.22mm, domain=-1.5:6,dashed,variable=\x,black!40!blue]  plot ({\x},{6.9});
		\draw[scale=1, line width=.26mm, domain=5.035:6,variable=\x,black!60!blue]  plot ({\x},{6.9});
		\draw[scale=1, line width=.22mm, domain=0:6.9,dashed,variable=\y,black!40!blue]  plot ({6},{\y});
		\draw[scale=1, line width=.22mm, domain=-1.5:3,dashed,variable=\x,black!40!blue]  plot ({\x},{1.65});
		\draw[scale=1, line width=.26mm, domain=2.035:3,variable=\x,black!60!blue]  plot ({\x},{1.65});
		\draw[scale=1, line width=.22mm, domain=0:1.65,dashed,variable=\y,black!40!blue]  plot ({3},{\y});	
		\draw[scale=1, line width=.22mm, domain=-0.5:5,dashed,variable=\x,black!40!blue]  plot ({\x},{4.5});
		\draw[scale=1, line width=.26mm, domain=4.035:5,variable=\x,black!60!blue]  plot ({\x},{4.5});
		\draw[scale=1, line width=.22mm, domain=0:4.587,dashed,variable=\y,black!40!blue]  plot ({5},{\y});	
		\draw[scale=1, line width=.22mm, domain=-0.5:2,dashed,variable=\x,black!40!blue]  plot ({\x},{.769});
		\draw[scale=1, line width=.26mm, domain=1.035:2,variable=\x,black!60!blue]  plot ({\x},{.769});
		\draw[scale=1, line width=.22mm, domain=0:.737,dashed,variable=\y,black!40!blue]  plot ({2},{\y});
		\draw[scale=1, line width=.26mm, domain=3.035:4,variable=\x,black!60!blue]  plot ({\x},{2.937});
		\draw[scale=1, line width=.26mm, domain=0.035:1,variable=\x,black!60!blue]  plot ({\x},{.1837});
		\draw[<->, line width=.25mm] (2,-0.3) -- (5,-0.3) node at (3.5,-0.65) {\LARGE $x_{l,b} - x_{l,a}$};		
		\draw[<->, line width=.25mm] (3,-1.3) -- (6,-1.3) node at (4.5,-1.65) {\LARGE $x_{l,d} - x_{l,c}$};		
		\draw[<->, line width=.25mm] (-0.5,0.9) -- (-0.5,4.4) node[rotate=90] at (-1,3) {\LARGE$c_{l,b} - c_{l,a}$};
		\draw[<->, line width=.25mm] (-1.5,1.75) -- (-1.5,6.8) node[rotate=90,] at (-2,4.4) {\LARGE $c_{l,d} - c_{l,c}$};									
		\end{tikzpicture}
		\caption{Marginally increasing piecewise constant bid prices}\label{fig:marginally_inc}
	\end{center}
\end{figure}

We note that the analogue of Theorem~\ref{thm:conditions_on_bids} holds for continuous bids and for strictly convex bid curves. This could also be seen as the limiting case, by taking the limit where the incerement $m$ goes~to~$0$.

We illustrate the conditions in Theorem~\ref{thm:conditions_on_bids}  by revisiting Example \ref{ex:first_simple_example} from Section \ref{sec:2}.

\begin{example}[Simple Market]
	Revisiting Example \ref{ex:first_simple_example} and Example \ref{ex:core_illust}, we observe that the bid from bidder $3$ does not satisfy the conditions in Theorem~\ref{thm:conditions_on_bids}, because the bid price for $400$ MW is not submitted. Assume instead bidder $3$ provides the mutually exclusive bids of $\$300$ for $400$ MW and $\$600$ for $800$ MW. Suppose bidders $1$ and $2$ change their bids to 
	$\$0$ for $400$ MW. Then, bidders $1$ and $2$ are the winners and each receives a payment of $\$300$. If they were multiple identities of a single bidder, after shill bidding, sum of their payments would decrease from $\$700$ to $\$600$.
\end{example}

Next, we consider a more general setting with different types of goods where it is still possible to derive conditions to ensure supermodularity and coalition-proof VCG outcomes.

\subsection{Markets for different types of goods}

We now consider reverse auctions where the central operator is procuring a set of different types of goods. Each bidder has a private true cost function $c:\mathbb R_+^t \rightarrow \mathbb R_+$ that is nondecreasing with $c(0)=0$. We assume that this cost has an additive from, $c(x)=\sum_{\tau=1}^{t}c_\tau(x_\tau)$. Typically, in these markets, bids are submitted separately for each type with an upper-bound on the amount to be procured, $\bar{X}_\tau\in\R_+$ \cite{abbaspourtorbati2016swiss}. The operator treats these bids as bids from different identities, and then distributes the payments accordingly. In this case, the set~$L$ is the extended set of bidders such that the bid profile $\mathcal C$ is given by the bids of the form $c_l:\mathbb R_+ \rightarrow \mathbb R_+$, $\bar{X_l}\in\R_+$,  for all $l\in L$.

Let $[t]= \{1,\ldots,t\}$. Define the set $\{A^\tau\}_{\tau=1}^t$ to be a partition of the set $L$ where each set $A^\tau\subseteq L$ is the set of bidders submitting a bid for goods of type $\tau$. Specifically, we consider auctions cleared by the optimization problem:
\begin{equation}	\label{eq:multiple}
\begin{split} 
J(\mathcal{C}_S) =&  \min_{x\in\R^{|L|}_+  } {\ \sum_{l \in S}c_l( x_l)} \\
&\ \ \mathrm{s.t. } \ \sum_{l\in A^T} x_l \geq M(T) ,\, \forall T \subseteq [t], \\
&\quad \quad\ \  x_l \leq \bar{X}_l,\, \forall l\in L, \\
&\quad \quad\ \ x_l = 0,\, \forall l\in L\setminus S, 
\end{split}
\end{equation}
where $A^T=\bigcup_{\tau\in T}A^\tau$. We minimize the sum of declared costs subject to constraints on the subsets of types $[t]$. Here, the function $M:2^{[t]}\to\R_+$ defines the amount the operator wants to procure from possible combinations of different types of goods. We assume that $M(\emptyset)=0$ (normalized). We remark that the optimization problem in~\eqref{eq:multiple} contains the case in Example  \ref{ex:second_simple_example}. 

In the following result we see that if $c_l$ and $M$  satisfy convexity and supermodularity conditions respectively, then $J$ is supermodular.

\begin{theorem}\label{thm:re2_proof}
	{The objective function $J$ given by \eqref{eq:multiple} is supermodular when} $c_l$ is increasing and convex for all $l \in L$, $M$ is supermodular 
and nondecreasing.
\end{theorem}

The proof is relegated to Appendix~\ref{app:F}, and it builds upon recent advances in polymatroid optimization.\footnote{Polymatroid is a polytope associated with a submodular function. We highlight that under the conditions in Theorem~\ref{thm:re2_proof}, the first set of constraints in~\eqref{eq:multiple} is a contra-polymatroid~\cite{schrijver2003combinatorial}.} As a corollary of this result, the conditions on $c_l$, $l\in L$ and $M$ in Theorem~\ref{thm:re2_proof} imply core VCG outcomes for \eqref{eq:multiple}. The conclusions of Theorem~\ref{thm:no_collusions} on shill bidding and collusion are further corollaries of this~result.

Next, we illustrate that the VCG outcome may not lie in the core for a general polyhedral constraint set by revisiting Example~\ref{ex:second_simple_example}, and then we illustrate how the conditions in Theorem~\ref{thm:re2_proof} imply core outcomes {in this example.}

\begin{example}[Power Market]
	Revisiting Example~\ref{ex:second_simple_example}, for the constraint set \eqref{newc2}, we highlight that $M$ is not supermodular: $M(\{\text{A},\text{B},\text{C}\}) + M(\{\text{A}\}) < M(\{\text{A},\text{B}\}) + M(\{\text{A},\text{C}\}) $ where $M(\{\text{A}\})=0$ and $M(\{\text{A},\text{B},\text{C}\})=0$. Then, the VCG outcome under collusion is blocked by a deal between the operator and bidder~$1$. {Now instead consider the following constraint set:}
	\begin{equation*}	\label{newcz}
	\begin{split}
	\{ x\in\{0,100\}^5  \,\rvert\, & x_1 + x_2 + x_3 \geq \tilde{M}(\{\text{A},\text{B}\})=100,\\
	& x_1 + x_4 +  x_5 \geq \tilde{M}(\{\text{A},\text{C}\})=100,\\
	& x_1 + x_2 + x_3 + x_4 + x_5 \geq \tilde{M}(\{\text{A},\text{B},\text{C}\})=200\}.  
	\end{split}
	\end{equation*}	
	Under this new constraint set, type A can still replace types B and C, but it cannot replace both types simultaneously. In other words, types B and C cannot complement each other to replace type~A as well.
	Note that the function $\tilde{M}$ is supermodular, normalized and nondecreasing, which satisfies all the requirements of Theorem~\ref{thm:re2_proof}.  {So, here the VCG outcome lies in the core.} Specifically,
	under the VCG mechanism, bidders $2$ and $4$ become winners and receive $p_2 = 350 + (650-600)= \$400$ and $p_4 = 250 + (750-600)= \$400$. This outcome is not blocked by any other coalition and collusion is not profitable for bidders. This example illustrates that marginally increasing cost curves alone are not enough to conclude the supermodularity of the reverse auction objective function in~(\ref{eq:main_abstraction}). 
\end{example}

\section{Core-Selecting Mechanisms}\label{sec:4}
\vspace{.1cm}
{In this section, we investigate further payment rules that are coalition-proof without any restrictions on the bidders and the constraint set. In particular, we show that any mechanism that selects its payments from the core is coalition-proof. Under such mechanisms, losing bidders cannot profit from a joint deviation, and a shill bidder cannot profit more than its truthful VCG utility.}

 We first define the core with respect to the submitted bids and refer to it as the revealed core.
\begin{definition}\label{def:r_core_def}
	For every set of bidders $R\subseteq L$, the revealed core $Core(\mathcal{B}_R)\in\R\times\R^{\rvert R\rvert}_+$ is defined as follows
	\begin{equation}\label{eq:corerevdef}
	Core(\mathcal{B}_R)=\Big\{\bar u\in\R\times\R^{\rvert R\rvert}_+ \,|\, \bar u_0+\sum\limits_{l\in R}\bar u_l=-J(\mathcal{B}_R),\,
	\bar u_0+\sum\limits_{l\in S}\bar u_l\geq-J(\mathcal{B}_S),\, \forall S \subset R \Big\}.
	\end{equation}	
\end{definition}

Next, we define utilities with respect to the submitted bids. The revealed utility of bidder~$l$ is defined by $\bar u_l(\mathcal B)=p_l(\mathcal B)-b_l(x^*_l(\mathcal B))$, and the revealed utility of the operator is defined by $\bar u_0(\mathcal B)=u_0(\mathcal B) =-\sum_{l\in L} p_l(\mathcal B) - d(x^*(\mathcal B),y^*(\mathcal B))$.

A mechanism is said to be \textit{core-selecting} if  it is selecting its payments such that the revealed utilities lie in the revealed core $Core(\mathcal{B})$. Then, the payment rule is given by
\begin{equation*}
p_l(\mathcal{B})=b_l(x^*_l(\mathcal{B}))+\bar u_{l}(\mathcal{B}),
\end{equation*} 
where $\bar u\in Core(\mathcal{B})$. We highlight that the revealed core can be defined without the true costs. As a remark, the pay-as-bid mechanism is a core-selecting mechanism where $\bar u_l=0$ for all $l\in L$. Our main result of this section shows that these mechanisms give rise to coalition-proof outcomes. 

\vspace{.1cm}
\begin{theorem}\label{thm:no_collusions_bocs}
	Consider a core-selecting auction mechanism modeled by \eqref{eq:main_abstraction}.
	\begin{itemize}
		\item[(i)] {A group of bidders who lose when bidding their true values cannot profit by a joint deviation.}
		\item[(ii)] \hspace{-.05cm}{Bidding with multiple identities is unprofitable for all bidders with respect to the VCG utilities.}\end{itemize}
\end{theorem}
\vspace{.1cm}

The proof is relegated to Appendix~\ref{app:G}. This proof can be used as an alternative approach to prove Theorem~\ref{thm:no_collusions}. We remark that the proof method differs from Theorem~\ref{thm:no_collusions} since this proof does not require supermodularity.

{The VCG mechanism is known to be the only dominant-strategy incentive-compatible efficient mechanism \cite{green1979incentives}; however, the VCG mechanism can be subject to collusion and shill bidding~\cite{yokoo2004effect}. We showed that this occurs since the VCG outcomes may not lie in the core and thus incentive-compatibility property of the VCG mechanism is relaxed under the core-selecting mechanisms~\cite{day2008core}.}  To alleviate this issue, we investigate the core-selecting mechanisms that minimize bidders' maximum gain from a unilateral deviation from truthful bidding.
We soon see that these mechanisms select the revealed utility allocations that are most preferable for the bidders.

A revealed utility allocation $\bar u\in Core(\mathcal B)$ is \textit{bidder-Pareto-optimal} if there is no $\tilde{u}\in Core(\mathcal B)$ such that $\tilde{u}_{l}\geq \bar u_{l}$ for each $l\in L$, and $\tilde{u}_{l}> \bar u_{l}$ for some bidder $l\in L$. In Figure~\ref{fig:core_picture}, the set of bidder-Pareto-optimal points correspond to the line segment of the core with the maximum total payment. Nash equilibria of the pay-as-bid mechanism for the market model~(\ref{eq:main_abstraction}) are also given by the bidder-Pareto-optimal points in the core with respect to the true costs, see~\cite[Theorem~1]{orcun2018game}. 

It is shown that a core-selecting mechanism minimizes the tendencies to deviate from truthful bids, among all other core-selecting mechanisms,\text{ if and only if} the mechanism chooses a bidder-Pareto-optimal revealed utility allocation \cite{day2008core}. This result follows from the fact that the maximum gain from a deviation from truthful bidding is given by the difference between the VCG payment and the core-selecting one\cite{day2007fair}. We call such mechanisms \textit{bidder optimal core-selecting (BOCS) mechanisms}. 
From Figure \ref{fig:core_picture}, we observe that there are many utility allocations satisfying this property. To obtain a unique outcome, we select the one that minimizes the Euclidean distance to the VCG utilities~\cite{day2012quadratic}. 
For a study on alternative bidder optimal core-selecting payment rules, we refer to \cite{erdil2010new,lubin2015new}.

We remark that if the VCG outcome lies in the core, then it is the unique bidder-Pareto-optimal point in the core, see our work in \cite{orcun2018game}. 
As a result, the VCG and BOCS mechanisms are equivalent under the supermodularity condition. If the supermodularity condition holds, truthful bidding also becomes the dominant-strategy Nash equilibrium of the BOCS mechanism. In Table~\ref{tab:compare_all}, the comparisons of the BOCS, VCG and pay-as-bid mechanisms are provided. 

\begin{table}[ht]
	\vspace{.2cm}
	\caption{Comparison of the mechanisms for the reverse auction \eqref{eq:main_abstraction} ($\star$~\textit{stands for Yes if supermodularity holds})}
	\label{tab:compare_all}
	\begin{center}
		\begin{tabular}{|l||l||l||l|}
			\hline
			Property & Pay-as-bid & BOCS & VCG\\
			\hline
			Incentive-compatibility & No & $\star$ & Yes\\
			\hline
			Utilities are in the revealed core & Yes & Yes & $\star$\\
			\hline
			Coalition-proofness & Yes & Yes & $\star$ \\
			\hline
		\end{tabular}
	\end{center}
\vspace{.2cm}
\end{table} 

Next, we investigate methods to calculate the BOCS payments. Using our previous results, we first reduce the problem size significantly. However, this calculation may still be computationally infeasible for some instances. To this end, we formulate this problem with an iterative approach, which converges fast in practice.

\subsection{Computing payments under the BOCS mechanism}

Finding a Pareto-optimal outcome is computationally difficult for auctions involving many bidders, because one needs to solve the reverse auction problem \eqref{eq:main_abstraction} for $2^{|L|}$ different subsets to define the core constraints in~\eqref{eq:corerevdef}. Furthermore, the reverse auction problem \eqref{eq:main_abstraction} can be NP-hard in some cases. Invoking Lemma~\ref{lem:lemma_core}, we can reduce the number of constraints to $2^{\rvert W\rvert}$, which grows exponentially only in the number of winners. We call this approach with the reduced number of core constraints \textit{the direct BOCS approach}. Here, we formulate this approach for calculating the BOCS payments.

We start by removing the operator from the definition. By invoking Lemma~\ref{lem:lemma_core}, we keep only the constraints corresponding to the subsets of the winners.  As a result, we have that $\bar u \in Core(\mathcal B)$, if and only if $\bar u_0= -J(\mathcal B)-\sum_{l\in L}\bar u_l$ and $\bar u_{-0}$ lies in
\begin{equation}\label{eq:simp_core}
\mathcal{K}(\mathcal B)=\Big\{\bar u_{-0}\in\R^{\rvert L\rvert}_+\,|\,\sum_{l \in K } \bar u_l\leq J(\mathcal{B}_{-K}) - J(\mathcal{B}), \forall K\subseteq W\Big\}.
\end{equation}

As we already discussed, there can be multiple bidder-Pareto-optimal core points. We first maximize the sum of the revealed utilities of the bidders, since such revealed utilities are bidder-Pareto-optimal, and they are also known to minimize the sum of maximum profits of all bidders by a deviation from truthful bidding~\cite{day2007fair}. We can also argue that adopting a total payment maximizing core-selecting mechanism promotes fairness since it is in the combined best interest of the bidders. 
Consequently, BOCS mechanisms maximize the revealed utility of bidders, as follows:
\begin{equation}	\label{eq:max_core}
\nu =  \max_{\substack{\bar u_{-0}\in\mathcal{K}(\mathcal B)} } \bm 1^\top \bar u_{-0}\,.
\end{equation}
From Figure \ref{fig:core_picture}, we observe that the solution $\bar u_{-0}^*$ to \eqref{eq:max_core} may not be unique since there are many bidder-Pareto-optimal utility maximizing allocations.  From this set, we select the point that minimizes the Euclidean distance to the VCG outcome,
\begin{equation}	\label{eq:quadratic_payment}
\bar u_{-0}^{\text{BOCS}} =  \argmin_{\substack{\bar u_{-0}\in\mathcal{K}(\mathcal B),\\ \bm 1^\top \bar u_{-0} = \nu } } \norm{\bar u_{-0}-\bar u_{-0}^{\text{VCG}}}_2^2\,,
\end{equation}
where $\bar u_{l}^{\text{VCG}}=J(\mathcal B_{-l})-J(\mathcal B),$ for all $l\in L$.
A winning bidder gets the payment $p_l^{\text{BOCS}}(\mathcal B)=b_l(x^*_l(\mathcal B))+\bar u_{l}^{\text{BOCS}}$. Taking the VCG revealed utilities as a reference rule is intuitive and one can choose another utility maximizing allocation by studying the market properties and its participants~\cite{day2012quadratic}. 

Unfortunately, this approach may still not be a computationally feasible one since there can be many winners to the reverse auction~\eqref{eq:main_abstraction}. Therefore, we study iterative approaches where core constraints are generated on demand.

\subsection{Iterative approach via core constraint generation}
Because the number of the core constraints in~\eqref{eq:simp_core} increases exponentially with the number of winners, we may not be able to enumerate all the core constraints for the reverse auction in \eqref{eq:main_abstraction}. Hence, the BOCS payments could be hard to calculate.
As suggested in \cite{day2007fair}, the state of the art approach for calculating a core outcome is to use constraint generation and in practice, this algorithm requires the generation of only several core constraints. The method was initially used in the 1950s in order to solve linear programs that have too many constraints \cite{dantzig1954solution}. Instead of directly solving the large problem, one solves a primary problem with only a subset of its original constraints. From this primary solution, one can formulate a secondary problem that adds another constraint to the first step. The algorithm iterates between these two problems and converges to the optimal solution of the large problem. 

Next, we formulate the core constraint generation algorithm to calculate the BOCS payments for the reverse auction \eqref{eq:main_abstraction}. Previously, this constraint generation method was only formulated for forward combinatorial auctions~\cite{day2007fair}. We take into account the conceptual differences of our reverse auction model in \eqref{eq:main_abstraction}. 

Bidders' revealed utilities at the first step of our algorithm are given by $\bar u^0_{-0}=\bar u^{\text{VCG}}_{-0}$, where $\bar u^{\text{VCG}}_{-0}$ is the revealed utilities of the VCG mechanism. Convergence of the algorithm does not require this choice. However, this choice is intuitive since $\bar u^{\text{VCG}}_{-0}$ is the solution to \eqref{eq:quadratic_payment} if the VCG outcome lies in the core. 

As an iterative method, at each step $k$, we find the blocking coalition that has the largest violation for the revealed utility allocation~$\bar u^k_{-0}$.\footnote{We remark that the coalition $C$ is a blocking coalition if $J(\mathcal{B}_C)+\sum_{l\in C}\bar u_l<-\bar u_0$, see the constraints in \eqref{eq:corerevdef}. } If a blocking coalition exists, the coalition with the largest violation for the revealed utility allocation~$\bar{u}^k_{-0}$ is given by \begin{equation}\label{eq:simpleform}C^k=\argmin_{C\subseteq L}\ J(\mathcal{B}_C)+\sum_{l \in C } \bar u_l^k.\end{equation} This follows from the constraints in $Core(\mathcal B)$ in \eqref{eq:corerevdef}. Let $W$ be the set of winners. It is straightforward to see that the existence of this~blocking~coalition~$C^k$ is equivalent to the violation of the constraint $\sum_{l \in W\setminus C^k} \bar u_l^k\leq J(\mathcal  B_{L\setminus\{W\setminus C^k\}})- J(\mathcal  B)$ from the set in \eqref{eq:simp_core}. This follows from the equivalent characterization of the core in Lemma~\ref{lem:lemma_core}. We call the set $W\setminus C^k$ the blocked winners, and the problem \eqref{eq:simpleform} generates the core constraint on the revealed utilities of this set of bidders.


Next, we reformulate the problem in \eqref{eq:simpleform} using inflated bids. Given a revealed utility allocation $\bar u_{-0}^k\in\R^{\rvert L\rvert }_+$, we define the central operator's objective at step~$k$ as 
\begin{equation*}
\bar{J}^k(x,y;\mathcal B)=\sum\limits_{l\in L} b_l^{\bar u_l^k}(x_l) + d(x,y),
\end{equation*}
where the inflated bid $b_l^{\bar u_l^k}(x_l)\in\R_+$ is given by
\begin{equation*}
b_l^{\bar u_l^k}(x_l) = 
\begin{cases}
0 & x_l=0\\
b_l(x_l)+\bar u_l^k& \text{otherwise}. 
\end{cases}
\end{equation*}
Note that even if the bid $b_l$ is a convex bid curve,  the inflated bid $b_l^{\bar u_l^k}$ is not convex if $\bar u_l^k\neq 0$, because of the discontinuity at $0$. As a remark, the discontinuity at $0$ can be described by binary variables.
Then, the optimization problem \eqref{eq:simpleform} for finding the blocking coalition with the largest violation is reformulated as follows
\begin{equation}\label{eq:main_algorithm_ccg}
\begin{split}
z(\bar u^k_{-0})&=\min_{x,y}\ \bar{J}^k(x,y;\mathcal B)\ \mathrm{s.t.}\ g(x,y)\leq 0, \\
x^*(\bar  u^k_{-0})&=\argmin_x \left\{\min_{\substack{y:\ g(x,y)\leq 0}}\bar{J}^k(x,y;\mathcal B) \right\}, \\
C^k&=\{l\in L\,\rvert\,  x_l^*(\bar  u^k_{-0})\neq0\},\\
\end{split}
\end{equation}
where $C^k$ is the blocking coalition with the largest violation for the revealed utility allocation~$\bar{u}^k_{-0}$. Notice that the problem \eqref{eq:main_algorithm_ccg} essentially solves the reverse auction where winners' bids are inflated by their revealed utilities from the earlier step. 

After obtaining the blocking coalition and the corresponding central operator cost $z(\bar{u}^k_{-0})$, we solve the following two problems to obtain another candidate for a BOCS revealed utility allocation. First, we take the subset of revealed utilities that are maximizing the total utility of bidders as follows
\begin{equation}	\label{eq:step_of_ccg}
\begin{split} 
\nu^k =&  \max_{\substack{\bar u_{-0}\in\R^{\rvert L\rvert }_+} } \bm 1^\top \bar u_{-0} \\
&\ \ \ \ \mathrm{s.t. } \sum_{l \in W\setminus{C^t} } \bar u_l\leq J(\mathcal B_{L\setminus\{W\setminus{C^t}\}})- J(\mathcal B),\, \forall t\leq k,\\
& \quad\quad\quad\ \bar u_{-0}\leq\bar u_{-0}^{\text{VCG}},\\
\end{split}
\end{equation}
where $J(\mathcal B_{L\setminus\{W\setminus{C^t}\}}) =z(\bar u_{-0}^t)-\sum_{l \in W\cap C^t } \bar u_l^t $. We highlight that this term, $J(\mathcal B_{L\setminus\{W\setminus{C^t}\}})$, does not require any further solution to the reverse auction problem.
Moreover, the last constraint in~\eqref{eq:step_of_ccg} follows from the core constraints since the VCG outcome is the maximal revealed utility in the core for every bidder~\cite{orcun2018game}. As a result, it is straightforward to see that the problem \eqref{eq:step_of_ccg} essentially solves the problem \eqref{eq:max_core} with only a subset of its original constraints.
Then, we determine a new candidate for the bidders' revealed utilities, $\bar{u}_{-0}^{k+1}$, as the solution to the following quadratic program:
\begin{equation}	\label{eq:next_step_of_ccg}
\begin{split} 
\bar{u}_{-0}^{k+1} =&  \argmin_{\substack{\bar u_{-0}\in\R^{\rvert L\rvert }_+ } } \norm{\bar u_{-0}-\bar u_{-0}^{\text{VCG}}}_2^2\\
&\  \ \ \mathrm{s.t. }\ \ \sum_{l \in W\setminus{C^t} } \bar u_l\leq J(\mathcal  B_{L\setminus\{W\setminus{C^t}\}})- J(\mathcal  B),\, \forall t\leq k,\\
&\ \ \ \quad\quad\quad \bar u_{-0}\leq\bar u_{-0}^{\text{VCG}},\  \bm 1^\top \bar u_{-0} = \nu^k. \\
\end{split}
\end{equation}

The entire process for determining the bidder optimal core-selecting payment is summarized in Algorithm \ref{alg:ccg_algorithm}.
\begin{algorithm}[H]
	\caption{Core Constraint Generation (CCG) Algorithm}
	\begin{algorithmic}[1]\label{alg:ccg_algorithm}
		\renewcommand{\algorithmicrequire}{\textbf{Initialize:} Solve the optimization problem \eqref{eq:main_abstraction}.}
		\renewcommand{\algorithmicensure}{  \textbf{Iteration step $k$:}                           }
		\REQUIRE 
		Calculate the VCG utilities and set  $\bar u^0=\bar u^{\text{VCG}}$ and $k=0$. \\
		\STATE Solve the optimization problem \eqref{eq:main_algorithm_ccg}. 
		\WHILE {$W\setminus{C^k}\neq\emptyset$}
		\STATE Obtain $\bar u_{-0}^{k+1}$ by solving \eqref{eq:step_of_ccg} and then \eqref{eq:next_step_of_ccg}.
		\STATE Update $k = k+1$.
		\STATE Solve the optimization problem \eqref{eq:main_algorithm_ccg}. 
		\ENDWHILE
	    \RETURN $\bar u_{-0}^{k}$.
	\end{algorithmic}
\end{algorithm}

Algorithm \ref{alg:ccg_algorithm} converges to the solution of the optimization problem \eqref{eq:quadratic_payment} for the BOCS mechanism. We refer to \cite[Theorem~4.2]{day2007fair} for the proof of convergence of this algorithm. We note that this algorithm may still require the generation of all possible core constraints, which is equivalent to solving the problem \eqref{eq:main_algorithm_ccg} $2^{\rvert W \rvert}$ times. In practice, even when there are many winners, the algorithm requires the generation of only several core constraints. Using this approach, we can significantly reduce the number of solutions needed for the optimization problem~\eqref{eq:main_abstraction} in order to calculate the BOCS payments.

\section{Numerical Results}\label{sec:5}
Our goal is to compare the effectiveness of the proposed mechanisms and methods based on electricity market data. Towards this we consider the pay-as-bid, the LMP, the BOCS, and the VCG mechanisms. First, we consider an optimal power dispatch problem with a four-node three-generator network and show that the VCG outcomes are not coalition-proof. Then, we illustrate the coalition-proof outcomes obtained by the BOCS mechanism. Second, we consider the IEEE test systems. We compare the total payment under the LMP, the BOCS, and the VCG mechanisms. We verify that the VCG outcomes are in the core for the considered instances. Then, with small modifications to the line limits, we show that the VCG outcomes do not necessarily lie in the core for these test systems as well. Finally, we study the two-stage Swiss reserve procurement auction \cite{abbaspourtorbati2016swiss}. We show that the VCG outcomes are not in the core, hence shill bidding and collusion can be profitable for the bidders.

In order to illustrate the applicability of the BOCS mechanism, we provide the wall-clock time for the computations performed using our setup.  We solve all optimization problems with GUROBI \cite{gurobi}, called through MATLAB via YALMIP~\cite{lofberg2005yalmip}, on a computer equipped with a 32 GB RAM and a 4.0 GHz quad-core Intel i7 processor. 
\subsection{Four-node three-generator network model}
We consider a dispatch problem with polytopic DC power flow constraints in Figure \ref{fig:three_node}, based on the models considered in \cite{wu1996folk}. Cost curves are quadratic polynomials. All lines have the same susceptance. In Figure \ref{fig:three_node}, line limit from node $i$ to node $j$ is denoted by $C_{i,j}=C_{j,i}\in\R_+$. Under the VCG mechanism, payments and utilities are given in the first column of Table~\ref{tableNET}. Suppose via coalition bidders $1$ and~$2$ change their bids to $b_l(x)=0$ for all $x\in\R_+$. Then, bidders $1$ and $2$ are the only winners of the dispatch problem and their payments and utilities are given in the second column of Table \ref{tableNET}. Collusion is profitable and the total payment of the operator increases from \$$260$ to \$$280$. 
\vspace{.2cm}
\begin{figure}[ht]
	\begin{center}
		\begin{tikzpicture}[scale=0.8, every node/.style={scale=0.5}]
		\draw[-,black!80!blue,line width=.32mm] (-0.3,0) -- (-0.3,0.42);
		\draw[-,black!80!blue,line width=.32mm] (0.3,0) -- (0.3,0.42);
		\draw[-,black!80!blue,line width=.32mm] (0.3,3.58) -- (0.3,4);
		\draw[-,black!80!blue,line width=.32mm] (-0.3,3.58) -- (-0.3,4);
		\draw[-,black!80!blue,line width=.32mm] (-3.59,1.8) -- (-4,1.8);
		\draw[-,black!80!blue,line width=.32mm] (-3.59,2.2) -- (-4,2.2);
		\draw[-,black!80!blue,line width=.32mm] (3.59,1.8) -- (4,1.8);
		\draw[-,black!80!blue,line width=.32mm] (3.59,2.2) -- (4,2.2);
		
		\draw[-,black!80!blue,line width=.3mm] (-3.6,1.8) -- (-0.3,0.4) node[anchor=west]  at(-4,0.25) {\LARGE $C_{3,2}=10\text{ MW}$};
		\draw[-,black!80!blue,line width=.3mm] (-0.3,3.6) -- (-3.6,2.2) node[anchor=west]  at(-4, 3.65) {\LARGE $C_{3,1}=10\text{ MW}$};
		\draw[-,black!80!blue,line width=.3mm] (0,4) -- (0,0) node[anchor=west]  at(0.1,2) {\LARGE $C_{1,2}=10\text{ MW}$};
		\draw[-,black!80!blue,line width=.3mm] (0.3,3.6) -- (3.6,2.2) node[anchor=west]  at(1.75,3.65) {\LARGE $C_{1,4}=10\text{ MW} $};
		\draw[-,black!80!blue,line width=.3mm] (0.3,0.4) -- (3.6,1.8) node[anchor=west]  at(1.75,0.25) {\LARGE $C_{2,4}=10\text{ MW} $};
		\draw[-,line width=.85mm] (-1,0) -- (1,0) node[anchor=west]  {\LARGE $2$};
		\draw[-,line width=.85mm] (-1,4) -- (1,4) node[anchor=west]  {\LARGE $1$};
		\draw[-,line width=.85mm] (-4,3) -- (-4,1) node[anchor=west]  {\LARGE $3$};
		\draw[-,line width=.85mm] (4,3) -- (4,1) node[anchor=west]  {\LARGE $4$};	
		\draw[<-,line width=.3mm] (0,-0.05) -- (0,-1);
		\draw[<-,line width=.3mm] (0,4.05) -- (0,5);
		\draw[<-,line width=.3mm] (-4.05,2) -- (-5,2);
		\draw[->,line width=.3mm] (4,2) -- (5,2);
		\draw (0,-1.5) circle (.5cm) node {\LARGE $G_2$} node at(2.7,-1.5) {\huge $c_2(x)=.1x^2 + 12x$};
		\draw (0,5.5) circle (.5cm)node {\LARGE $G_1$} node at(2.7, 5.6) {\huge $c_1(x)=.1x^2 + 12x$};
		\draw (-5.5,2) circle (.5cm) node {\LARGE $G_3$}node at(-8,2) {\huge $c_3(x)=.1x^2 + 5x$};
		\draw (5.5,2) circle (.5cm)node {\LARGE $D$} node at(7.6, 2) {\huge $D=20\text{ MWh}$};
		\end{tikzpicture}
		\caption{Four-node three-generator network}\label{fig:three_node}
	\end{center}
\end{figure}
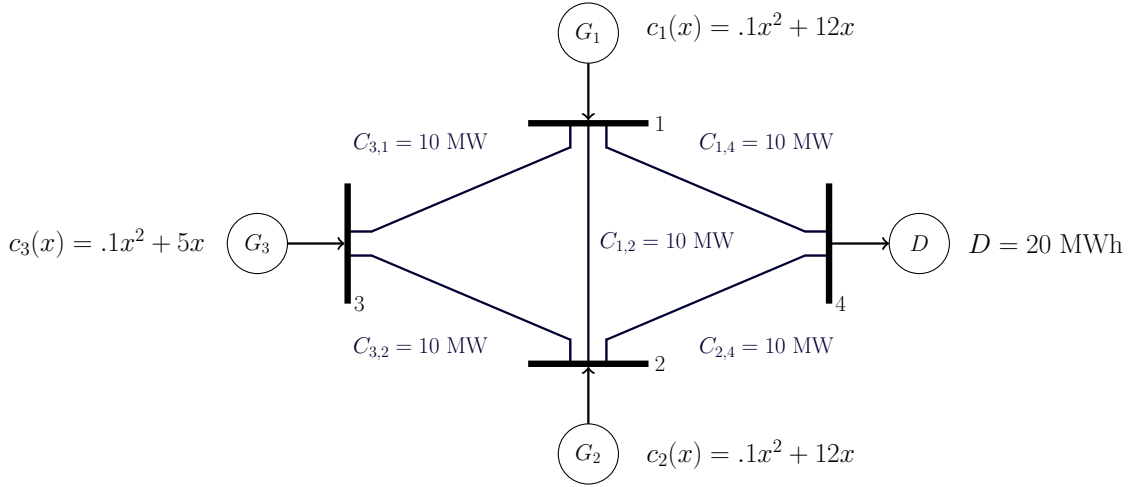

\begin{table}[h]
	\caption{The VCG outcomes for the network model (\$, MWh)}
	\label{tableNET}
	\begin{center}
		\begin{tabular}{|l||l||l||l||l|}
			\hline
			& \multicolumn{2}{l||}{Truthful Bidding} & \multicolumn{2}{l|}{Collusion ($1$, $2$)}\\
			\hline
			& payment (utility)  & $x$ & payment (utility)  & $x$  \\
			\hline
			Bidder $1$ & $0$ $(0)$ & 0 & $140$ $(10)$& 10 \\
			\hline
			Bidder $2$  & $0$  $(0)$ & 0 & $140$ $(10)$ & 10 \\
			\hline
			Bidder $3$ & $260$  $(120)$ & 20& $0$ $(0)$ & 0\\
			\hline
		\end{tabular}
	\end{center}
\end{table}

Next, we consider the BOCS mechanism. Under the BOCS mechanism, payments and utilities are given in the first column of Table~\ref{tableNET2}. Suppose, via coalition, bidders $1$ and $2$ change their bids to $b_l(x)=0$ for all $x\in\R^t_+$. Then, bidders $1$ and $2$ are the only winners of the dispatch problem and their payments and utilities are given in the second column of Table~\ref{tableNET2}. We observe that the collective utility of bidders $1$ and~$2$ reduces to $-\$120$. Hence, in this case, collusion is not profitable for bidders~$1$~and~$2$. Furthermore, after the collusion, the total payment of the operator also reduces from \$$260$ to \$$140$. This example shows that the core-selecting mechanisms can eliminate collusion of the losing bidders.

\begin{table}[h]
	\caption{The BOCS outcomes for the network model (\$, MWh)}
	\label{tableNET2}
	\begin{center}
		\begin{tabular}{|l||l||l||l||l|}
			\hline
			& \multicolumn{2}{l||}{Truthful Bidding} & \multicolumn{2}{l|}{Collusion ($1$, $2$)}\\
			\hline
			& payment (utility)  & $x$ & payment (utility)  & $x$  \\
			\hline
			Bidder $1$ & $0$ $(0)$ & 0 & $70$ $(-60)$& 10 \\
			\hline
			Bidder $2$  & $0$  $(0)$ & 0 & $70$ $(-60)$ & 10 \\
			\hline
			Bidder $3$ & $260$  $(120)$ & 20& $0$ $(0)$ & 0\\
			\hline
		\end{tabular}
	\end{center}
\end{table}
\subsection{IEEE test systems with DC power flow models}
The following simulations are based on the IEEE test systems with polytopic DC power flow constraints adopting the models considered in \cite{wu1996folk}.
\subsubsection{14-bus, 30-bus and 118-bus test systems}

We consider the IEEE $14$-bus \cite{christie2000power}, $30$-bus \cite{alsac1974optimal, ferrero1997transaction} and $118$-bus test systems \cite{christie2000power}. We assume all bidders are truthful and the true cost curves are convex quadratic polynomials, provided in the references. In practice, truthfulness can only hold under the VCG mechanism since it is dominant-strategy incentive-compatible. The corresponding total payments of the mechanisms are shown in Table \ref{tablenet1}. All the mechanisms lead to the same winner allocation as expected. 
\begin{table}[ht]
	\caption{Total payments of the IEEE test systems}
	\label{tablenet1}
	\begin{center}
		\begin{tabular}{|l||l||l||l|}
			\hline
			Mechanism & $14$-bus & $30$-bus & $118$-bus \\
			\hline
			Pay-as-bid & \$$7642.6$  & \$$565.2$  & \$$125947.8$  \\
			\hline
			LMP & \$$10105.1$  &  \$$716.9$  & \$$167055.8$  \\  
			\hline
			BOCS& \$$10513.4$ & \$$746.4$  & \$$169300.4$ \\  
			\hline
			VCG & \$$10513.4$  & \$$746.4$  & \$$169300.4$  \\
			\hline
		\end{tabular}
	\end{center}
\end{table}

For all three test systems, we observe that the VCG mechanism has a slightly larger total payment than the LMP mechanism. Moreover, the VCG payment of every bidder is larger than its LMP payment. For the DC power flow models with increasing convex bids, this result was proven in \cite[Theorem~2]{xu2017efficient}. 

Another observation is that the VCG outcomes are in the core for all systems. Next, we provide an explanation for each system. $14$-bus and $118$-bus systems do not have any line limits, hence, they have the form of \eqref{eq:simpler_clearing_model}. Invoking Theorem~\ref{thm:conditions_on_bids}, we conclude that supermodularity condition holds. Despite the fact that $30$-bus system has line limits, the VCG outcome is in the core. This result can be explained in two ways. First, none of the line limit constraints are tight. Second, we observe that removing two bidders can yield to an infeasible problem, similar to Proposition \ref{lem:removal_of_two}. These test systems are specialized instances and they do not necessarily conclude that the VCG mechanism is coalition-proof for the DC power flow models. We examine this shortcoming of the VCG mechanism in our next simulation.

Computation times are provided only for the $118$-bus case, because the other problems are trivially small. For this electricity market, the direct BOCS approach is not computationally feasible, because there are $19$ winners out of $54$ bidders and the optimal cost calculation takes $451$~milliseconds. This approach would require $66$ hours. Computation times for the VCG mechanism and the CCG algorithm are $24.8$ and $31.4$ seconds respectively. After the VCG mechanism, the CCG algorithm converges only in a single iteration. This iteration takes 6.6 seconds, because it involves binary variables whereas the optimal cost calculation for the market model does not. 

\subsubsection{Effect of line limits}

We consider the IEEE $14$-bus test system, with a line limit on lines exiting node $1$, connecting node~$1$ to nodes $2$ and $5$. We set this line limit to be $10$ MW.  We again assume all bidders are truthful. The corresponding total payments of the mechanisms are shown in Table \ref{tablenet2}. 

\begin{table}[ht]
	\caption{Total payments of the IEEE 14-bus test system with line limits}
	\label{tablenet2}
	\begin{center}
		\begin{tabular}{|l||l|}
			\hline
			Mechanism & 14-bus with line limits \\
			\hline
			Pay-as-bid & \$$9715.2$  \\
			\hline
			LMP & \$$10361.0$  \\  
			\hline
			BOCS& \$$11220.1$   \\  
			\hline
			VCG & \$$11432.1$  \\
			\hline
		\end{tabular}
	\end{center}
\end{table}

We observe that the VCG outcome does not lie in the core. Line limits are tight and the problem does not have the form of \eqref{eq:simpler_clearing_model}. Hence, shill bidding and collusion can be profitable for bidders. Moreover, we observe that the BOCS mechanism yields a larger total payment than the LMP mechanism. We underline that this does not necessarily hold for the payment of a single bidder, because the BOCS payment depends on the Pareto-optimal point chosen. For some bidders, the BOCS payment can be equal to the pay-as-bid payment (which is smaller than the LMP payment), whereas for some others it can be equal to the VCG payment (which is larger than the LMP payment). 

With this result, we also reiterate that convex bid curves are not enough to ensure that the VCG outcomes are in the core. Similar results can be obtained for $118$-bus test system by fixing $50$~MW limits on two lines, one connecting nodes~$5$ and~$6$, another connecting nodes~$9$ and~$10$.  

For the 14-bus example, computation times for the VCG mechanism and the CCG algorithm are $3.6$ and $8.2$ seconds respectively. After the VCG mechanism, the CCG algorithm converges in $4$~iterations. 

\subsection{Swiss reserve procurement auctions}

The following simulations are based on the bids placed in the 46th weekly Swiss reserve procurement auction of 2014~\cite{abbaspourtorbati2016swiss}. The reverse auction involves $21$ power plants bidding for secondary reserves, $25$ for  positive tertiary and 21 for negative tertiary reserves. There are complex constraints arising from nonlinear cumulative distribution functions. The~constraints imply that the deficit of reserves cannot occur with a probability higher than 0.2\%. Moreover, the constraints include coupling between first and second stage decision variables corresponding to the weekly and daily reserve auctions. The corresponding total payments and procured MWs of the pay-as-bid mechanism, the BOCS mechanism, and the VCG mechanism are shown in Table \ref{tablen3}. 
\begin{table}[ht]
	\caption{Total payments of the two stage auction }
	\label{tablen3}
	\begin{center}
		\begin{tabular}{|l||l||l||l|}
			\hline
			Type & SR & PTR & NTR \\
			\hline
			Procured MWs & $409$ MW  & $100$ MW & $114$ MW \\
			\hline
			\multicolumn{2}{|l||}{Total Pay-as-bid payment} & \multicolumn{2}{l|}{$2.293$ million CHF} \\
			\hline
			\multicolumn{2}{|l||}{Total BOCS payment}  & \multicolumn{2}{l|}{ $2.437$ million CHF} \\  
			\hline
			\multicolumn{2}{|l||}{Total VCG payment} & \multicolumn{2}{l|}{$2.529$ million CHF} \\
			\hline
		\end{tabular}
	\end{center}
\end{table}

While the VCG payment rule yields the highest total payment, the BOCS yields the second highest. Note that the bids from the power plants are not marginally nondecreasing for all quantities and the constraints are nonstandard. Hence, the supermodularity condition does not hold and we observe that the VCG outcome does not lie in the core. Consequently, shill bidding and collusion are profitable for bidders under the VCG mechanism.

For this electricity market, the direct BOCS approach is not computationally feasible, because there are $28$ winners and the optimal cost calculation takes $8$ seconds. This approach would require $68$ years. Computation times for the VCG mechanism and the CCG  algorithm are $580.6$ and $659.2$~seconds respectively. After the VCG mechanism, the CCG algorithm converges in $4$~iterations.

\section{Conclusion}\label{sec:6}
We introduced a constrained optimization problem to model reverse auctions that may involve continuous values of different types of goods, general nonconvex constraints, and second stage costs. We discussed the game theoretic analysis of these reverse auction mechanisms under different payment rules. We first showed that the VCG mechanism results in a dominant-strategy incentive-compatible Nash equilibrium. Through examples, we then showed that this mechanism suffers from collusion and shill bidding. Motivated by this problem, we derived three different conditions under which collusion and shill bidding are not profitable, and hence the VCG mechanism is coalition-proof. Since these conditions are restrictive and they may not capture the constrained optimization problem under consideration, we investigated the closest we can achieve to the property of incentive-compatibility under a coalition-proof mechanism. To this end, we formulated the bidder optimal core-selecting mechanism. By removing incentives for manipulations, we expect the bidding process to be simplified, and this can help promote participation in the market. Finally, we verified our results in several case studies based on electricity market data.

As a future work, we will explore learning Nash equilibria in such markets to model the behavior of the bidders in a repeated setting. As an extension, we will consider budget balance in double-sided auctions. Also, we will address the issue of pricing intermittent generation.
\section*{Acknowledgments}\label{sec:7}
We are grateful to Rico Zenklusen for helpful discussions on matroid theory. We thank Swissgrid for electricity market~data.

\appendix
\section{Proof of Theorem~\ref{thm:incentive_comp}}\label{app:A}
\begin{proof}
	(i) We distinguish between bidder $l$ placing a generic bid $\mathcal{B}_l = b_l$ and bidding truthfully $\mathcal C_l= c_l$.  For the set of bids $\mathcal{B}$, the utility of bidder $l$ is given by:
	\begin{align*}
	u_l({\mathcal{B}}) = h(\mathcal{B}_{-l})\, -\Big(\sum\limits_{k\neq l} b_k(x_k^*(\mathcal{B})) + c_l(x_l^*(\mathcal{B})) + d(x^*(\mathcal{B}),y^*(\mathcal{B}))\Big),
	\end{align*}
	where the term in brackets is the cost $\bar{J}$ of $\hat{\mathcal C} =(\mathcal C_l, \mathcal{B}_{-l})$ but evaluated at $(x^*(\mathcal{B}),y^*(\mathcal{B}))$, see \eqref{eq:main_abstraction}. For $\hat{\mathcal C}$ note that $ u_l(\hat{\mathcal C})= h(\mathcal{B}_{-l}) - J(\hat{\mathcal C})$. Then, we have the following:
	\begin{equation*}
	J(\hat{\mathcal C}) \leq \sum\limits_{k\neq l} b_k(x_k^*(\mathcal{B})) + c_l(x_l^*(\mathcal{B})) + d(x^*(\mathcal{B}),y^*(\mathcal{B})).
	\end{equation*}
	We can show that $u_l(\hat{\mathcal C}) \geq u_l({\mathcal{B}})$ because $(x^*(\mathcal{B}),y^*(\mathcal{B}))$ is a feasible suboptimal allocation for the auction under the bids $\hat{\mathcal C}$. Therefore, bidding truthfully is a best response strategy, regardless of other bidders' strategies $\mathcal{B}_{-l}$.
	
	(ii) By the definition of VCG payment rule and incentive-compatibility, we have $p_l(\mathcal C)= u_l(\mathcal C) + c_l(x_l^*(\mathcal C))$ where $\mathcal C=\{c_l\}_{l\in L}$. We then have: $u_0(\mathcal C) = - \sum_{{l\in L}} c_l(x_l^*(\mathcal C)) - d(x^*(\mathcal{C}),y^*(\mathcal{C})) - \sum_{{l\in L}}{u_l(\mathcal C)}$. The sum of utilities, $\sum_{l=0}^{\lvert L\rvert}  u_l(\mathcal C)=- \sum_{{l\in L}} c_l(x_l^*(\mathcal C)) - d(x^*(\mathcal{C}),y^*(\mathcal{C})) $ is maximized since $(x^*(\mathcal{C}),y^*(\mathcal{C}))$ is the minimizer to the optimization problem \eqref{eq:main_abstraction} under true costs.
	
	(iii) Nonnegative payments can be verified substituting Clarke pivot rule for $h(\mathcal{B}_{-l})$:
	\begin{equation*}
	p_l(\mathcal{B})=b_l(x^*_l(\mathcal{B}))+(J(\mathcal{B}_{-l})-J(\mathcal{B}))\geq0,
	\end{equation*}
	for all set of bids $\mathcal{B}$. For individual rationality, assume bidders are not bidding less than their true costs, that is, $b_l(x)\geq c_l(x),\, \forall x\in\R_+^t$. We have
	\begin{equation*}
	u_l(\mathcal{B}) = b_l(x_l^*(\mathcal{B}))-c_l(x_l^*(\mathcal{B}))+J(\mathcal{B}_{-l}) - J(\mathcal{B}) \geq 0,
	\end{equation*}
	for all set of bids $\mathcal{B}$.
	\QEDA
\end{proof}
\section{Proof of Lemma~\ref{lem:lemma_core}}\label{app:B}
\begin{proof}
The utility allocation
	$ u^{}$ is unblocked by every $S\subseteq L$ if and only if
	\begin{equation*}
	-J(\mathcal{C}_S) \leq \sum_{l \in S} u^{}_l + u_0 = \sum_{l \in S} u_l - \sum_{l \in W}  u^{}_l - J(\mathcal{C}),\,\forall S\subseteq L,
	\end{equation*}
	since the losing bidders are not allocated and they obtain zero payment. Thus, the core can equivalently be parametrized as $$\sum_{l \in W\setminus S} u_l \leq J(\mathcal{C}_S)- J(\mathcal{C}),\ \forall S\subseteq L.$$ Moreover, fixing the set $K = W\setminus S$, the dominant constraints are those corresponding to minimal $J(\mathcal{C}_S)$, in particular, when we have $S = L\setminus K$ (this being maximal set with $K$ not taking part in the coalition~$S$). \QEDA
\end{proof}

\section{Proof of Theorem~\ref{thm:iff_supermodularity}}\label{app:C}
\begin{proof}
We prove that supermodularity is sufficient for $u^{\text{VCG}}$ to lie in the core, $Core(\mathcal{C}_R)$, for all $R\subseteq L$. Notice that we have $u_{l,R}^{\text{VCG}}=J(\mathcal{C}_{R\setminus{l}}) - J(\mathcal{C}_R)$, we drop the dependence on $R$ for the sake of simplicity in notation. We recall Lemma~\ref{lem:lemma_core}. Hence, we need to show that
	\begin{equation}
	\label{eq:core_constraints2}
	\sum_{l \in K }  u_l\leq J(\mathcal{C}_{R \setminus K}) - J(\mathcal{C}_R),\ \forall K\subseteq W.
	\end{equation}
	
	Let $K=\{ l_1, \dots, l_k\}$. 
	{Notice that, by supermodularity,
		$
		J(\mathcal{C}_{R\setminus{l_\kappa}}) 
		- J(\mathcal{C}_R) 
		\leq 
		J(\mathcal{C}_{R \setminus \{ l_\kappa ,..., l_k \} })
		-
		J(\mathcal{C}_{R\setminus \{ l_{\kappa+1} ,..., l_k \}}) 
		$. 
		Thus 
		\begin{align*}
		\sum_{l\in K} u_l^{\text{VCG}} 
		&= 
		\sum_{\kappa=1}^k
		J(\mathcal{C}_{R\setminus{l_\kappa}}) 
		- J(\mathcal{C}_R) 
		\\
		&\leq 
		\sum_{\kappa =1}^k
		J(\mathcal{C}_{R \setminus \{ l_\kappa ,..., l_k \}} )
		- J(\mathcal{C}_{R\setminus  \{ l_{\kappa+1} ,..., l_k \}}) 
		\\
		&=
		J(\mathcal{C}_{R \setminus K} ) - J(\mathcal{C}_R).
		\end{align*}
		The last equality holds by a telescoping sum.}
	Thus, we see that \eqref{eq:core_constraints2} holds and so, by Lemma~\ref{lem:lemma_core}, the VCG outcome belongs to the core. The same argument can be repeated to obtained core VCG outcomes for any set of participating auction bidders $R\subseteq L$ and for any profile $\mathcal{C}$.

	To prove that supermodularity is also necessary for outcomes to lie in the core, {we proceed by contradiction.} Suppose that the supermodularity condition does not hold for a bidder~${l}$. Then, there exist sets ${S}\subseteq{R}$ where $J(\mathcal{C}_{R\setminus{l}})\! -\! J(\mathcal{C}_R)\!>\!  J(\mathcal{C}_{S\setminus{l}})\! -\! J(\mathcal{C}_S)$. We may,{ without loss of generality,} choose $R= S\cup \{i\}$ for some $i$. To see this, take $S^0=S$ and
	$S^\kappa=S^{\kappa-1}\cup \{ l_\kappa\}$ with $S^k=R$, then,
		\begin{align}
		&\sum_{\kappa =1}^k J(\mathcal{C}_{S^\kappa\setminus{l}})-J(\mathcal{C}_{S^{\kappa-1}\setminus{l}}) = J(\mathcal{C}_{R\setminus{l}})-J(\mathcal{C}_{S\setminus{l}}) \notag\\
		&\hspace{4.25cm} >  J(\mathcal{C}_R)-J(\mathcal{C}_S) = \sum_{\kappa=1}^k J(\mathcal{C}_{S^\kappa})-J(\mathcal{C}_{S^{\kappa-1}}).\notag
		\end{align}
	The strict inequality above must hold for one of the summands  $J(\mathcal{C}_{S^\kappa\setminus {l}})-J(\mathcal{C}_{S^{\kappa-1}\setminus {l}}) > J(\mathcal{C}_{S^\kappa})-J(\mathcal{C}_{S^{\kappa-1}})$.  So we may consider sets ${S} \subseteq {R}$ that differ by one bidder, say $i$. Thus, by this observation, we have
	$ J(\mathcal{C}_{{R}\setminus {l}}) - J(\mathcal{C}_{R}) > J(\mathcal{C}_{{S}\setminus {l}}) - J(\mathcal{C}_{S}) = J(\mathcal{C}_{{R}\setminus{\{i,l\}}}) - J(\mathcal{C}_{{R}\setminus {i}})$ for $i \in {R}\setminus {S}$.
	Further, after rearranging the above inequality we obtain
	{
		\begin{equation}\label{eq:JR}
		J(\mathcal{C}_{R\setminus{i}})-J(\mathcal{C}_{R}) > J(\mathcal{C}_{R\setminus{\{i,l\}}})-J(\mathcal{C}_{R\setminus{l}})\geq 0.
		\end{equation}	
	}That is, \textit{both} bidder $i$ and $l$ are winners of the VCG auction with bidders $R$.  
	Considering the auction with the set of bidders ${R}$ and with  $i,{l}\in W$, and ${K}=\{i, {l}\}$, we have: 
	{
		\begin{align*}
		\sum_{l' \in {K}} u_{l'}^{\text{VCG}}=&\sum_{l' \in {K}}\!\! {J(\mathcal{C}_{{R}\setminus{l'}}) - J(\mathcal{C}_{R})} \\
		=& 
		J(\mathcal{C}_{{R}\setminus{\{i\}}}) - J(\mathcal{C}_{{R}}) + J(\mathcal{C}_{{R}\setminus{l}}) - J(\mathcal{C}_{R}) 
		\\
		>& 
		J(\mathcal{C}_{{R}\setminus{\{i,l\}}}) - J(\mathcal{C}_{{R}\setminus{l}}) + J(\mathcal{C}_{{R}\setminus{l}}) - J(\mathcal{C}_{R}) 
		\\
		=& J(\mathcal{C}_{{R} \setminus {K}}) - J(\mathcal{C}_{R}),
		\end{align*}
		where in the inequality above we apply \eqref{eq:JR}.
		Thus, given Lemma~\ref{lem:lemma_core}, \eqref{eq:core_constraints2} does not hold, consequently, $u^{\text{VCG}} \notin Core(\mathcal{C}_R)$. Thus the outcome of the VCG mechanism is not in the core for the subset of bidders $R\subseteq L$.
	}
	\QEDA
\end{proof}
	
\section{Proof of Theorem~\ref{thm:no_collusions}}\label{app:D}
\begin{proof}
		(i) Let $K$ be a set of colluders who would lose the auction when bidding their true values $\mathcal{C}_K=\{c_l\}_{l\in K}$, when bidding $\mathcal{B}_K=\{b_l\}_{l\in K}$ they become winners, that is, they are all allocated a positive quantity. We define $ {\mathcal{C}} = (\mathcal{C}_{K}, \mathcal{C}_{-K})$ and ${\mathcal{B}}=({\mathcal{B}}_{ K},  \mathcal{C}_{-K})$ where $\mathcal{C}_{-K}=\{c_l\}_{l\in L\setminus K}$ denotes the bidding profile of the remaining bidders. As a remark, the profile $\mathcal{C}_{-K}$ is not necessarily a truthful profile. The VCG utility that each player $l$ in $K$ receives under ${\mathcal{B}}$ is
		\begin{align*} 
		u_l^{\text{VCG}}({\mathcal{B}}) &\leq u_l^{\text{VCG}}({\mathcal{B}}_{-l}, \mathcal{C}_l)\\  &= J({\mathcal{B}}_{-l} ) - J({\mathcal{B}}_{-l}, \mathcal{C}_l) \\  
		& \leq J( \mathcal{C}_{-K}) - J( \mathcal{C}_{-K} , {\mathcal{C}}_l) \\
		& =  J( \mathcal{C}_{-l}) - J( \mathcal{C}) \\ &=u_l^{\text{VCG}}({\mathcal{C}})\\&=0,
		\end{align*}
		where the first inequality follows from the dominant-strategy incentive-compatibility of the VCG mechanism. The first equality comes from the definition of the VCG mechanism and the second inequality applies the supermodularity of the function~$J$. The next equality comes from the fact that the set $K$ originally was a group of losers. 
		
		So we see that, for all $l \in K$, the utility $u_l^{\text{VCG}}({\mathcal{B}}) $ is upper bounded by the utility that bidder $l$ would receive if every colluder was bidding truthfully. However, by the initial assumption, these bidders were losers while bidding truthfully, and hence $u_l^{\text{VCG}}({\mathcal{C}})=0,$ for all $l \in K$. Thus there is no benefit for losers from colluding by jointly deviating from their truthful bids.
		
		(ii) Similar to part (i), define $\mathcal{C}=(\mathcal{C}_{-l},\mathcal{C}_{l})$. The profile $\mathcal{C}_{-l}$ is not necessarily a truthful profile. Shill bids of bidder~$l$ are given by $\mathcal{B}_{S}=\{b_k\}_{k\in S}$.~We define a merged bid $\tilde{\mathcal{B}}_l$ as $$\tilde{b}_l(x_l)=\min_{x_k\in\R_+^t,\,\forall k}\, \sum_{k\in S}b_k(x_k)\ \mathrm{s.t. }\sum_{k\in S}x_k=x_l.$$ We then define ${\tilde{\mathcal{B}}}=(\mathcal{C}_{-l},{\tilde{\mathcal{B}}}_{l})$.
		The VCG utility from shill bidding under ${\mathcal{B}}=(\mathcal{C}_{-l},{\mathcal{B}}_{S})$, $\sum_{k\in S}{u}_k^{\text{VCG}}({\mathcal{B}})$, is given by
		\begin{align*} &= \sum_{k\in S}[J({\mathcal{B}_{-k}})-J({\mathcal{B}})+b_k(x_k^*({\mathcal{B}}))]-{c}_l(\sum_{k\in S}x_k^*({\mathcal{B}}))\\
		&\leq [J({\mathcal{B}}_{-S} ) - J({\mathcal{B}})]+\sum_{k\in S} b_k(x_k^*({\mathcal{B}}))-{c}_l(\sum_{k\in S}x_k^*({\mathcal{B}}))\\ 
		&= [J({\mathcal{C}}_{-l} ) - J({\tilde{\mathcal{B}}})]+\tilde{b}_l(\sum_{k\in S}x_k^*({\mathcal{B}}))-{c}_l(\sum_{k\in S}x_k^*({\mathcal{B}}))\\ 
		&=  {u}_l^{\text{VCG}}({\tilde{\mathcal{B}}}) \\ 
		&\leq  {u}_l^{\text{VCG}}({\mathcal{C}})
		\end{align*}
		The first inequality follows from the supermodularity of~$J$. The second equality holds since we have $J({\tilde{\mathcal{B}}})=J({{\mathcal{B}}})$. This follows from the definition of the merged bid and the following implication. Since the goods of the same type are fungible for the central operator, the functions $g$ and~$d$ in fact depend on $\sum_{l\in L} x_l$. 
		The third equality follows from the definition of VCG utility. The second inequality is the dominant-strategy incentive-compatibility of the VCG mechanism. Therefore, the total VCG utility that $l$ receives from shill bidding is upper bounded by the utility that $l$ would receive by bidding truthfully as a single bidder. Making use of shills, hence, is not profitable. 

\QEDA
	\end{proof}
\section{Proof of Theorem~\ref{thm:conditions_on_bids}}\label{app:E}
To prove Theorem~\ref{thm:conditions_on_bids}, the following lemma is needed.

\begin{lemma}\label{lem:increasing_quantities}
	Under the market model \eqref{eq:simpler_clearing_model}, for an auction with bidders $S$ and $R=S\cup\{j\}$  with corresponding allocations  $x$ and $x'$,  marginally increasing costs imply that 
	$\forall l\in S, \
	x_l'\leq x_l.
	$
\end{lemma}

\begin{proof}
	The proof  follows by contradiction. That is, we will show that when $x'$ is such that ${x}'_{{l}} > x_{{l}}$, for some $l\in S$, then $x'$ can be modified to provide a lower cost via the allocation $q$ for bidders~$S$ (thus contradicting optimality of $x$). First, we notice that since bids are equally spaced by the amount $m$ and with marginally increasing cost,  $\sum_{l \in S}x_l = \sum_{l \in R}{x'_l} = M$ holds. Now, in order to procure exactly $M$ amount from bidders $R$, some bidders' allocations must decrease, that is, the set $K = \{ l \in S \,\rvert\, x_l'< x_l  \}$ is nonempty. Consider a feasible allocation $q'$ for the auction with bidders $R$ where the amount $M$ is being procured and
	\[
	q'_l = 
	\begin{cases}
	x'_j, &\text{ for }l=j,\\
	x_l, & \text{ for } l\in S\backslash K, \\
	q'_l, &\text{ for }l\in K \text{ where }x_l' \leq q'_l \leq x_l.
	\end{cases}
	\]
	So, $q'$ is constructed from $x'$ by transferring the amount $m'~=~\sum_{l\in S\backslash K} x'_l - x_l $ from bidders in $S\backslash K$ to bidders in $K$. In doing so, the inequality $x'_l \leq q_l' \leq x_l$ can be maintained:
	\[
	m' \leq \sum_{l\in K} (x_l -x'_l).
	\]
	The above inequality holds because when summing over $l\in S$, $x_l$'s sum to $M$ and $x'_l$'s sum to $M-x_j$.
	
	Since {$x'$} is optimal for bidders $R$ and $q'$ is not:
	\begin{align}
	J(\mathcal{C}_R) =& c_j{(x'_j)} + \sum_{l\in S\backslash K } c_{{l}}{(x'_l)} + \sum_{l \in K }{c_l}{(x'_l)}\notag\\
	\leq & c_j{(x'_j)} + \sum_{l\in S\backslash K} c_l{(x_l)} + \sum_{l\in K} c_l{(q'_l)}=\bar{J}(q'),  \label{eq:q'}
	\end{align}
	where we used $\bar{J}(q')$ as a short-hand-notation for the cost corresponding to choosing the allocation~$q'$ under truthful bidding. 
	
	Now, we use the marginally increasing true costs to replace the summations over $K$ in \eqref{eq:q'}. 
	In particular, define $q=(q_l : l\in S)$ so that 
	$$ 
	q_l := x_l +(x'_l- {q'_l})=
	\begin{cases}
	x'_l & \text{for } l \in S \backslash K,\\
	x_l + x'_l -q'_l & \text{for } l \in K.
	\end{cases}
	$$
	Note that $q$ is feasible since $x'$ and $q'$ have the same sum over $S$ (and thus cancel) and $x_l$ is feasible. 
	Further, since $(q_l - x_l)  = (x'_l- {q'_l})$,
	\begin{equation} 
	\label{eq:increas}
	\sum_{l \in K}{ c_l{(q_l)}  - c_l{({x'_l})}} < \sum_{l \in K}{c_l{(x_l)} - c_l{({q'_l})}}.
	\end{equation}
	Adding \eqref{eq:increas} to both side of \eqref{eq:q'} (and canceling $c_j{(x'_j)}$) gives
	\begin{align*}
	\bar{J}(q) &= \sum_{l\in S\backslash K} c_l{(x'_l)} + \sum_{l\in K} c_l{(q_l)}\\
	& < \sum_{l\in S\backslash K} c_l{(x_l)} + \sum_{l\in K} c_l{(x_l)} = J(\mathcal{C}_S),
	\end{align*}
	which contradicts the optimality of $x$. \QEDA
\end{proof}
\begin{corollary}
	Given the conditions in Theorem~\ref{thm:conditions_on_bids} and the optimal allocation to procure the amount $M$, for any lower amount $\tilde M\leq M$(while being multiple of $m$)  to be procured, the allocation for each bidder does not increase.
\end{corollary}

This corollary follows from Lemma~\ref{lem:increasing_quantities}. Now, we are ready to prove Theorem~\ref{thm:conditions_on_bids}. 

\begin{proof}
	We prove that $J$ is supermodular. We adopt the same notation used in Lemma~\ref{lem:increasing_quantities} and we identify with $W\subseteq S$ the set of winners.
	For each $l \notin W $, we have by definition $x_{l} = 0$. 
	Thus $0=J(\mathcal{C}_{S\setminus{l}})-J(\mathcal{C}_S)$, (since the optimal solution is unchanged when $l$ is removed from S). By Lemma~\ref{lem:increasing_quantities},  ${x}'_{l}=0$ and so $0=J(\mathcal{C}_{R\setminus{l}})-J(\mathcal{C}_R)$ also. Thus, supermodularity holds for $l \notin W$.
	
	For each winning bidder $w \in W$, denote $u_{w}^{\text{VCG}}(\mathcal{C}_S) = J(\mathcal{C}_{S\setminus{w}})- J(\mathcal{C}_S)$. Adopting the same notation of Lemma~\ref{lem:increasing_quantities}, we can indicate it as: 
	\[
	u_{w}^{\text{VCG}}(\mathcal{C}_S):= - c_{w}{(x_{w})}+\sum_{l \in S_{-w}} { ( c_l{(\varrho_l)} - c_l{(x_l)} ) }, 
	\]
	where $\varrho_l$ are the optimal allocations of each $l \in S_{-w}$, when $w$ exits the auction. By Lemma~\ref{lem:increasing_quantities}, $\varrho_l\geq x_l$. Similarly, after bidder $i$ enters the auction, $u_{w}^{\text{VCG}}(\mathcal{C}_R) = J(\mathcal{C}_{R\setminus{w}})- J(\mathcal{C}_R)$. That is,
	\begin{equation*}
	u_{w}^{\text{VCG}}(\mathcal{C}_R):=- c_{w}{({x}'_{w})}  + \sum_{l \in S_{-w}} { ( c_l{(\varrho_l')} - c_l{({x'_l})} ) } +  c_i{(\varrho_i')} - c_i{({x'_i})},
	\end{equation*}
	where $\varrho_i'$ are the amounts accepted from $ l \in R_{- w}$ when $w$ exits the new auction. By Lemma~\ref{lem:increasing_quantities}, we again have $\varrho_l'\geq {x'_l}$.
	
	Notice that so far we applied Lemma~\ref{lem:increasing_quantities} to justify the increase of the accepted amounts, first, from each $l\in S_{- w}$ and now from $l \in R_{-w}$, due to the exit of $w$ from the auctions. We can apply Lemma~\ref{lem:increasing_quantities} again and affirm that $ x'_l \leq x_l\ \forall l \in S$, and in particular $ x'_{w} \leq x_{w}$, because of the entrance of $i$.
	
	We now find suitable lower and upper bounds to ensure  inequality  $u_{w}^{\text{VCG}}(\mathcal{C}_R) \leq u_{w}^{\text{VCG}}(\mathcal{C}_S)$. First, note that 
	$ J(\mathcal{C}_S)= \sum_{ l \in S_{-w}}{ c_l{(x_l)} } + c_{w}{(x_{w})} \leq \sum_{ l \in S_{-w}}{ c_l{(q_l)} } + c_{w}{(x'_{w})}$, 
	where $q_l$'s are from the cheapest allocation to procure the amount $(M - x'_w)$ among $S_{-w}$.
	By Lemma~\ref{lem:increasing_quantities} we  have $\varrho_l \geq q_l \geq x_l \hspace{1em} \forall l \in S_{-w}$, since $x_l$'s sum to $(M - x_w) \leq (M- x'_w) $ (due to $x'_w \leq x_w$), and $\varrho_l$'s sum to $M$. 
	Moreover, since every $c_l$ is marginally increasing, $q_l$'s are such that $\sum_{l \in S_{-w}} ({\varrho_l}- q_l) = x'_{w}$, because exactly the amount $M$ is purchased.
	Using the above suboptimal allocation, we  have a lower bound for $u_{w}^{\text{VCG}}(\mathcal{C}_S)$:  
	\begin{equation}
	\label{lower_bound}
	u_{w}^{\text{VCG}}(\mathcal{C}_S) \geq \sum_{l \in S_{-w}} { ( c_l{(\varrho_l)} - c_l{(q_l)} ) } - c_{w}{(x'_{w})} .
	\end{equation}
	Defining now $\delta_l =({\varrho_l} - q_l), \forall l \in S_{-w}$ we must have $\sum_{ l\in S_{-w}} \delta_l = x'_{w}$ and $ J(\mathcal{C}_{R\setminus{w}})= \sum_{l \in S_{-w}} {c_l{(\varrho_l')} } + c_i{(\varrho_i')} \leq   \sum_{l \in S_{-w}} {c_l{({x'_l} + \delta_l)} } + c_i{(x'_i)}$ , since the right hand side is a feasible cost to procure the amount $M$ among the bidders $\{S,i\}\setminus w$. Indeed, $\sum_{l \in S}{x'_l} + x'_i = M$ and $\sum_{l \in S_{-w}} {\delta_l} = x'_{w}$. Hence, we have: 
	\begin{equation*}
	\label{upper_bound}
	\begin{split}
	u_{w}^{\text{VCG}}(\mathcal{C}_R) \leq \sum_{l \in S_{-w}}& {(c_l{(x'_l + \delta_l)} - c_l{(x'_l)})} + (c_i{(x'_i)} - c_i{(x'_i)}) - c_{w}{(x'_{w})} .
	\end{split}
	\end{equation*}
	Moreover, via marginally increasing costs (also via strictly convex costs), we have: 
	\begin{equation} 
	\label{increasing_marginal}
	( c_l{(x'_l + \delta_l)} - c_l{(x'_l)} ) \leq  ( c_l{(\varrho_l)} - c_l{(q_l)} ),\, \forall l \in S_{-w}.
	\end{equation}
	The above holds because $\forall l \in S_{-w}$, $(\varrho_l - q_l) =( x'_l + \delta_l - x'_l) = \delta_l$ and $x'_l \leq q_l$. In particular, $x'_l$ are the amounts accepted to procure the amount $(M-x'_{w})$ among $\{S,i\} \setminus w$, while $q_l$ are those to procure the same amount among $S_{-w}$. Then, combining equations~\eqref{increasing_marginal} and \eqref{lower_bound}, we finally obtain $u_{w}^{\text{VCG}}(\mathcal{C}_R) \leq u_{w}^{\text{VCG}}(\mathcal{C}_S).$ As a result, we obtain supermodularity and this concludes the proof. \QEDA
\end{proof}

\section{Proof of Theorem~\ref{thm:re2_proof}}\label{app:F}
\begin{proof}
	We prove that $J$ is supermodular. To this end, we reparametrize the problem \eqref{eq:multiple} into another class of optimization problem. We define the function $f:2^{ L }\to\R_+$ as follows
	\begin{equation*}
		f(A)=\max_{\substack{\forall T \subseteq \{1,\ldots,t\}, \\ A\supseteq A^T}}\, M(T).
	\end{equation*}
	
	We consider the following optimization problem:
	\begin{equation}	\label{eq:multiple_type}
	\begin{split} 
	J(\mathcal{C}_S) =&  \min_{x\in\R^{|L|}_+  } {\ \sum_{l \in S}c_l( x_l)} \\
	&\ \ \mathrm{s.t. } \ \sum_{l\in A} x_l \geq f(A) ,\, \forall A \subseteq L, \\
	&\quad \quad\ \  x_l \leq \bar{X}_l,\, \forall l\in L, \\
	&\quad \quad\ \ x_l = 0,\, \forall l\in L\setminus S.
	\end{split}
	\end{equation}
	
	First, notice that $f(A^T)=M(T)$. Since $x_l\geq0,\,\forall l$, the constraints added by the definition of the function $f$ are all redundant constraints. Then, these constraints in \eqref{eq:multiple_type} are feasible, once the constraints in \eqref{eq:multiple} are satisfied. Hence, these problems are equivalent.  We also remark that $f(\emptyset)=0$ and $f$ is nondecreasing.
	
	Before we proceed, we need to show that supermodularity of the function $M$ implies supermodularity of the function $f$. Suppose $f(A)=M(T^A)$ and $f(B)=M(T^B)$ for some $T^A$, $T^B$. Then,
	\begin{equation*}
		\begin{split}
		f(A)+f(B)=\,&M(T^A)+M(T^B)\\
					\leq\,&M(T^A\cup T^B)+M(T^A\cap T^B)\\
					\leq\,&f(A\cup B)+f(A\cap B).
		\end{split}
	\end{equation*}
	First inequality follows from supermodularity of the function~$M$. The last inequality holds since these sets are feasible suboptimums for $f(A\cup B)$ and $f(A\cap B)$. We conclude that the function $f$ is supermodular.
	
	Note that, given supermodularity of the function $f$, the first set of constraints in \eqref{eq:multiple_type} defines a \textit{contra-polymatroid}, see \cite[Section~44]{schrijver2003combinatorial} for a detailed exposition. {This class of problems are important in combinatorial optimization because they can often be solved in polynomial time. For the proof, we extend the work in \cite{he2012polymatroid} on replenishment games to reverse auctions over contra-polymatroids and box constraints.}
	
	We first show that the constraint  $\sum_{l\in L} x_l = f(L)$ is redundant and can be added to the original constraint set. We denote $S^{c}=L\setminus S$ and denote $x^*$ as the optimal allocation for \eqref{eq:multiple_type}. 
	It can be shown that for every $x_k^*$, $k\in S$ , there exists a set $k\in A_k \subseteq L$ such that  $\sum_{l\in A_k} x^*_l = f(A_k)$ is tight at the optimal solution. We can prove this via contradiction. Assume, for  $x_k^*$, $k\in S$, there does not exist a set $k\in A_k\subseteq L$ such that $\sum_{l\in A_k} x^*_l \geq f(A_k)$ is tight. Then, one can simply decrease the value of $x^*_k$ and get a lower objective value. Furthermore, note that any constraint corresponding to $A\not\supset S^c$ is redundant to $A\cup S^c$ because $f$ is nondecreasing and $x_l = 0$, for all $l\in S^c$. Then, this set $A_k$ has to be a superset of $S^c$, $A_k \supset S^c $.
	
	Next, we show that if the constraints for $A$ and $B$ are tight, so is the constraint for $A\cup B$. 
	\begin{subequations}
		\begin{align}
		f(A\cup B) + f(A\cap B) &\geq f(A) + f(B)\label{11a} \\ &= \sum_{l\in A} x^*_l + \sum_{l\in B} x^*_l\label{11b} \\ &= \sum_{l\in A\cup B} x^*_l + \sum_{l\in A\cap B} x^*_l\label{11c} \\ &\geq f(A\cup B) + f(A\cap B). \label{11d}
		\end{align}
	\end{subequations}	
	Inequality~\eqref{11a} follows from supermodularity of $f$, $f(A\cup B)  - f(A) \geq f(B) - f(A\cap B)$. Equality~\eqref{11b} follows from $A$ and $B$ being tight. We arranged the terms in the equality~\eqref{11c}. Inequality~\eqref{11d} follows from the feasibility of $x^*$ for the problem in~\eqref{eq:multiple_type}. Then, it is easy to see that  \eqref{11d} is in fact an equality and we can conclude that $\sum_{l\in A\cup B} x^*_l = f(A\cup B) $. 
	
	Recall that the constraint corresponding to the set $ A_l\supset S^c\cup \{l\}$ is tight for $x_l^*$. Hence, we can conclude that the constraint corresponding to the set $\bigcup_{l\in S}A_l = L$ is also tight and $\sum_{l\in L} x_l^* = f(L)$ holds for optimal solution.
	
	Next, we reformulate the first set of constraints in \eqref{eq:multiple_type} as follows.
	\begin{equation}\label{eq:ref1}
	\mathcal{P} = \Big\{ x\in\R_+ \,\rvert\,  -f(A) \geq \sum_{l\in A^c} x_l -\sum_{l\in L} x_l ,\, \forall A \subseteq L\Big\}.
	\end{equation}	
	Define $h(A)=-f(A^c)$ where $h(\emptyset)= -f(L)= -\sum_{l\in L} x_l $ and reorganize the constraint set~\eqref{eq:ref1}. 
	\begin{equation*}\label{eq:ref}
	\begin{split}
	\mathcal{P} &= \Big\{ x\in\R_+ \,\rvert\, h(A^c) +\sum_{l\in L} x_l \geq \sum_{l\in A^c} x_l ,\, \forall A \subseteq L\Big\}\\
	& = \Big\{ x\in\R_+ \,\rvert\,  h(A) +\sum_{l\in L} x_l \geq \sum_{l\in A} x_l ,\ ,\forall A \subseteq L\Big\}\\
	&  = \Big\{ x\in\R_+ \,\rvert\,  h(A) - h(\emptyset) \geq \sum_{l\in A} x_l ,\, \forall A \subseteq L,\, h(\emptyset) = -\sum_{l\in L} x_l  \Big\}. 
	\end{split}	
	\end{equation*}	
	We see that $k(A)=h(A)-h(\emptyset)$ is a nondecreasing submodular function and it is normalized. In literature, the function  $k$ is called a \textit{rank function}\cite{schrijver2003combinatorial}. Note that, $k(S)=h(S)-h(\emptyset)=-g(L\setminus S)+g(L)$. From the feasibility of the problem $\eqref{eq:multiple_type}$, we have that $g(L\setminus S)=0$ and $k(S)=\sum_{l\in L} x_l  $. We reorganize the constraint set  and obtain the following:
	\begin{equation*}
	\mathcal{P} = \Big\{ x\in\R_+  \,\rvert\, \sum_{l\in A} x_l \leq k(A) ,\, \forall A \subseteq L,\, k(S) = \sum_{l\in L} x_l  \Big\}. \\
	\end{equation*}	
	Finally, given that $k$ is nondecreasing and $x_l = 0,\, \forall l\in L\setminus S$, we can show that the constraints corresponding to $A\supset S$ are all redundant upper bounds. Then, we obtain the following
	\begin{equation*}	
	\mathcal{P} = \Big\{ x\in\R_+  \,\rvert\, \sum_{l\in A} x_l \leq k(A) ,\, \forall A \subseteq S,\, k(S) = \sum_{l\in S} x_l  \Big\}.
	\end{equation*}	
	The following result is known, and we refer to \cite[Theorem~6.1]{yao1997stochastic} and \cite{fujishige1980lexicographically}. The set $\mathcal{P}\cap \{x \,\rvert\,   x_l \leq \bar{X}_l,\, \forall l\}$ is equivalent to the set,
	\begin{align*}
	&\mathcal{P}'= \Big\{ x\in\R_+  \,\rvert\, \sum_{l\in A} x_l \leq \bar f(A) ,\, \forall A \subseteq S,\,k(S) = \sum_{l\in S} x_l  \Big\}, 
	\end{align*}	
	where $\bar f(A)=\min_{B\subseteq A}\{ k(A\setminus B)+\sum_{l\in B}\bar{X}_l  \}$ and this function is also a rank function \cite{yao1997stochastic}. We also assert that $\bar f(S)=k(S)$. To see that:
	\begin{equation*}\label{eq:deriv}
	\begin{split}
	k(S)=\sum_{l\in S} x_l &= \sum_{l\in S\setminus B} x_l +\sum_{l\in B} x_l\\
	&\leq k(S\setminus B)+\sum_{l\in B}\bar{X}_l,\, \forall B\subseteq S, \\
	\end{split}	
	\end{equation*}	
	hence, $\bar f(S)=\min_{B\subseteq S}\{ k(S\setminus B)+\sum_{l\in B}\bar{X}_l  \}= k(S)$.
	Finally, we obtain the following
	\begin{equation}	\label{eq:final_form}
	\begin{split} 
	-J(\mathcal{C}_S) =&  \max_{\substack{x_l\in\R_+, \forall l \in S}}  \,\sum_{l \in S}-c_l( x_l) \\
	&\ \ \quad  \mathrm{s.t. } \  \ \sum_{l\in A} x_l \leq \bar f(A) ,\, \forall A \subseteq S, \\
	&\quad\quad\quad\ \ \sum_{l\in S} x_l=\bar f(S). \\
	\end{split}
	\end{equation} 
	
	The first set of constraints in \eqref{eq:final_form} defines a \textit{polymatroid}, see \cite[Section 44]{schrijver2003combinatorial}. In \cite[Theorem~3]{he2012polymatroid}, it is proven that maximizing a separable concave function over a polymatroid results in a submodular objective function. The result also directly includes optimizing over the base polymatroid where the polymatroid constraint set is intersected with the equality $\sum_{l\in S} x_l=\bar f(S)$. Then, invoking this result, we conclude that $-J$ is submodular, and $J$ is supermodular. \QEDA
\end{proof}
\section{Proof of Theorem~\ref{thm:no_collusions_bocs}}\label{app:G}
First, we need the following lemma.
\begin{lemma}\label{lem:implicore}
	Let $\bar u\in\R\times\R^{\rvert L\rvert}_+$ be a revealed utility allocation in $Core(\mathcal{B})$. Then, for every set of bidders $K\subseteq L$ we have $\sum_{l\in K}\bar{u}_l({\mathcal{B}})\leq J({\mathcal{B}}_{-K} ) - J({\mathcal{B}})$.
\end{lemma}
\begin{proof} Since $\bar u_0=-J(\mathcal{B})-\sum_{l\in L}\bar u_l$, we reorganize the inequality constraint as follows $-J(\mathcal{B})-\sum_{l\in L\setminus S}\bar u_l\geq-J(\mathcal{B}_S),\, \forall S \subseteq L$.
	Setting $K=L\setminus S$ yields the statement.  \QEDA
\end{proof}

Next, we prove that core-selecting mechanisms are coalition-proof.
\begin{proof}
	(i) Let $K $ be a set of colluders who would lose the auction when bidding their true values $\mathcal C_l=c_l$, when bidding ${\mathcal{B}}_l= b_l $ they become winners, that is, they are all allocated a positive quantity. We define $\hat{\mathcal{C}} = (C_{K},\mathcal{B}_{-K})$ and ${\mathcal{B}}=({\mathcal{B}}_{ K}, \mathcal{B}_{-K})$ where $\mathcal{B}_{-K}=\{b_l\}_{l\in L\setminus K}$ denotes the bidding profile of the remaining bidders. As a remark, the profile $\mathcal{B}_{-K}$ is not necessarily a truthful or a strategic profile. We denote the \text{utility} that each bidder $l$ receives as $u_l$. The total utility that colluders receive under ${\mathcal{B}}$ is
	\begin{align*} 
	\sum_{l\in K}{u}_l({\mathcal{B}}) &= \sum_{l\in K}\bar{u}_l({\mathcal{B}})+b_l(x_l^*({\mathcal{B}}))-{c}_l(x_l^*({\mathcal{B}}))\\
	&\hspace{-0.25cm}\leq J({\mathcal{B}}_{-K} ) - J({\mathcal{B}})+\sum_{l\in K} b_l(x_l^*({\mathcal{B}}))-{c}_l(x_l^*({\mathcal{B}}))\\ 
	&\hspace{-0.25cm} = J(\hat{\mathcal{C}}) - \Big[\sum_{l\in L} b_l(x_l^*({\mathcal{B}})) + d(x^*({\mathcal{B}}),y^*({\mathcal{B}})) -\sum_{l\in K} b_l(x_l^*({\mathcal{B}}))+\sum_{l\in K}{c}_l(x_l^*({\mathcal{B}}))\Big] \\
	&\hspace{-0.25cm}  = J(\hat{\mathcal{C}}) - \Big[\sum_{l\in K} {c}_l(x_l^*({\mathcal{B}})) + \sum_{l\in L\setminus K} {b}_l(x_l^*({\mathcal{B}})) +d(x^*({\mathcal{B}}),y^*({\mathcal{B}}))\Big] \\ &\hspace{-0.25cm}\leq 0 = \sum_{l\in K}{u}_l({\hat{\mathcal{C}}}).
	\end{align*}
	The first equality follows from the core-selecting payment rule, where $\bar{u}({\mathcal{B}})$ is the revealed utility allocation. {The first inequality follows from Lemma~\ref{lem:implicore}.} The second equality comes from the fact that the set $K$ originally was a group of losers, {so $J(\mathcal{B}_{-K})=J(\hat{\mathcal{C}})$}. 
	After substituting these terms, we see that  the term in brackets is the cost $\bar{J}$ of $\hat{\mathcal{C}}$ but evaluated at a feasible suboptimal allocation $(x^*(\mathcal{B}),y^*(\mathcal{B}))$. Then, $\sum_{l\in K}{u}_l({\mathcal{B}}) $ is upper bounded by $0$ which is the total utility that the colluders would receive bidding truthfully. 
	
	As a result, there is at least one colluder not facing any benefit in collusion. Moreover, they cannot increase their collective utility by a joint deviation. Hence, collusion is not profitable for the losing bidders.
	
	(ii) Define $\mathcal{C}=(\mathcal{C}_{-l},\mathcal{C}_{l})$, where $\mathcal{C}_{-l}$ denotes the bidding profile of the remaining bidders. The profile $\mathcal{C}_{-l}$ is not necessarily a truthful profile. Shill bids of bidder~$l$ are given by $\mathcal{B}_{S}=\{b_k\}_{k\in S}$.~We define a merged bid $\tilde{\mathcal{B}}_l$ as $$\tilde{b}_l(x_l)=\min_{x_k\in\R_+^t,\,\forall k}\, \sum_{k\in S}b_k(x_k)\ \mathrm{s.t. }\sum_{k\in S}x_k=x_l.$$ We then define ${\tilde{\mathcal{B}}}=(\mathcal{C}_{-l},{\tilde{\mathcal{B}}}_{l})$.
	The total utility obtained from shill bidding under ${\mathcal{B}}=(\mathcal{C}_{-l},{\mathcal{B}}_{S})$, $\sum_{k\in S}{u}_k({\mathcal{B}})$, is given by
	\begin{align*} &= \sum_{k\in S}[\bar{u}_k({\mathcal{B}})+b_k(x_k^*({\mathcal{B}}))]-{c}_l(\sum_{k\in S}x_k^*({\mathcal{B}}))\\
	&\leq [J({\mathcal{B}}_{-S} ) - J({\mathcal{B}})]+\sum_{k\in S} b_k(x_k^*({\mathcal{B}}))-{c}_l(\sum_{k\in S}x_k^*({\mathcal{B}}))\\ 
	&= [J({\mathcal{C}}_{-l} ) - J({\tilde{\mathcal{B}}})]+\tilde{b}_l(\sum_{k\in S}x_k^*({\mathcal{B}}))-{c}_l(\sum_{k\in S}x_k^*({\mathcal{B}}))\\ 
	&=  {u}_l^{\text{VCG}}({\tilde{\mathcal{B}}}) \\ 
	&\leq  {u}_l^{\text{VCG}}({\mathcal{C}})
	\end{align*}
	The first inequality follows from the core-selecting payment rule and Lemma~\ref{lem:implicore}. The second equality holds since we have $J({\tilde{\mathcal{B}}})=J({{\mathcal{B}}})$. This follows from the definition of the merged bid and the following implication. Since the goods of the same type are fungible for the central operator, the functions $g$ and~$d$ in fact depend on $\sum_{l\in L} x_l$. 
	The third equality follows from the definition of the VCG utility. The second inequality is the dominant-strategy incentive-compatibility of the VCG mechanism. Therefore, the total utility that $l$ receives from shill bidding is upper bounded by the utility that $l$ would receive by bidding truthfully as a single bidder in a VCG auction. Making use of shills, hence, is not profitable with respect to the VCG utilities. 
\QEDA
\end{proof}

\bibliographystyle{IEEEtran}
\bibliography{IEEEabrv,library}
\end{document}